\documentclass[11pt]{article}
\pdfoutput=1
\usepackage[left=1in,right=1in,top=.9in,bottom=.9in]{geometry}

\usepackage{amsthm}

\usepackage[pagebackref,colorlinks]{hyperref}
    \hypersetup{linkcolor=blue!50!black, 
        citecolor=black!60, 
        urlcolor=blue!50!black
        }
        
\usepackage{amsmath}
\usepackage{bbm}  
\usepackage{mathrsfs}
\usepackage[nameinlink]{cleveref}
    \crefname{ex}{Example}{Examples}
    \crefname{thm}{Theorem}{Theorems} 
    \crefname{lem}{Lemma}{Lemmas}
    \crefname{prop}{Proposition}{Propositions}
    \crefname{cor}{Corollary}{Corollaries} 
    \crefname{conj}{Conjecture}{Conjectures} 
    \crefname{defn}{Definition}{Definitions}
    \crefname{const}{Construction}{Constructions}
    \crefname{rmk}{Remark}{Remarks}
    \crefname{lin}{Line}{Lines} 
    \crefname{prob}{Problem}{Problems}

\theoremstyle{plain}
	\newtheorem{thm}{Theorem}[section]
	\newtheorem{lem}[thm]{Lemma}
	\newtheorem{prop}[thm]{Proposition}
	\newtheorem{cor}[thm]{Corollary}
	
        \newtheorem{cons}[thm]{Construction}
	\newtheorem*{thm*}{Theorem}
	\newtheorem*{cor*}{Corollary}
	\theoremstyle{definition} 
		\newtheorem{defn}[thm]{Definition}
		\newtheorem{ex}[thm]{Example}
    	\newtheorem{rmk}[thm]{Remark}

\usepackage{tikz, graphicx, pgfplots}
\usepackage[margin=1.5cm]{caption}
\usepackage{subcaption}
\usetikzlibrary{positioning}
\usetikzlibrary{arrows}
\usetikzlibrary{decorations.pathmorphing}
    \tikzset{%
    fwdrxn/.style={very thick, arrows={-Stealth[length=5pt,width=5pt]}},
    revrxn/.style={very thick, arrows={-Stealth[length=5pt,width=5pt,left]}},
    newt/.style={turq, opacity=0.15}
    }
\usepackage{color, xcolor}

    \definecolor{ratecnst}{RGB}{172,172,172}
	
\usepackage{multirow}
\usepackage{array}

\usepackage{amsfonts, amssymb,mathtools}
\usepackage[alphabetic]{amsrefs}
\newcommand{\eq}[1]{\begin{align*}#1\end{align*}}
	\newcommand{\eqn}[1]{\begin{align}#1\end{align}}  
\newcommand{\ds}{\displaystyle}	
  
\newcommand\mc[1]{\mathcal{#1}}

\newcommand\mrm[1]{\mathrm{#1}}
\newcommand{\cc}{\ensuremath{\mathbb{C}}} 
\newcommand{\rr}{\ensuremath{\mathbb{R}}}   
 
\newcommand{\zz}{\ensuremath{\mathbb{Z}}}

\renewcommand{\epsilon}{\varepsilon}	
\renewcommand{\phi}{\varphi}			

\newcommand{\vv}[1]{{\boldsymbol{#1}}}  
\newcommand{\mm}[1]{\mathbf{#1}}               
\newcommand{\rrp}{\rr_{\geq 0}}
\newcommand{\rrpp}{\rr_{>0}}
\newcommand{\zzp}{\zz_{\geq 0}}
\newcommand{\zzpp}{\zz_{>0}}
\newcommand{\RR}{\ensuremath{\rightleftharpoons}}
\newcommand{\xx}{\vv x}
\newcommand{\yy}{\vv y}
\newcommand{\kk}{{k}}
\newcommand{\ratecnst}[1]{{\footnotesize{\color{gray}{#1}}}}

\usepackage{enumitem}
\usepackage{chemfig} \newcommand{\cf}[1]{\textsf{#1}}

\DeclareMathOperator{\codim}{codim}

\DeclareMathOperator{\spn}{span}
\DeclareMathOperator{\supp}{supp} 

\DeclareMathOperator{\row}{row} 
\DeclareMathOperator{\Int}{int} 
\renewcommand{\top}{{}}

\usepackage{comment}
\usepackage[colorinlistoftodos,prependcaption,
    textsize=footnotesize,  
    linecolor=gray, bordercolor=gray, backgroundcolor=gray!25, 
    textcolor=black]{todonotes}

\usepackage{authblk}\usepackage[symbol]{footmisc}

\title{Maximum likelihood estimation of log-affine
models \\ using detailed-balanced reaction networks}

\author[1]{
    Oskar~Henriksson%
}\author[2]{
    Carlos~Am\'endola%
}\author[3]{
    Jose~Israel~Rodriguez%
}\author[4]{
    Polly~Y.~Yu%
}
\affil[1]{\small Department of Mathematical Sciences, University of Copenhagen}
\affil[2]{\small Institute of Mathematics, Technical University of Berlin}
\affil[3]{\small  Department of Mathematics and 
Department of Electrical and Computer Engineering, University~of~Wisconsin -- Madison}
\affil[4]{\small Department of Mathematics, University of Illinois Urbana-Champaign}
\date{}

\usepackage{colordvi}
\usepackage{color}
\newcommand{\defcolor}[1]{{\color{blue}#1}}
\newcommand{\demph}[1]{\defcolor{{\it #1}}}

\usepackage{nicefrac}
\usepackage{wasysym}

\numberwithin{equation}{section}

\newcommand{\Xac}{\mc X_{\mm A,\vv c}}
\newcommand{\closedXac}{\overline{\mathcal{X}}_{\mm A,\vv c}}
\newcommand{\Xa}{\mc X_{\mm A}}
\newcommand{\closedXa}{\overline{\mathcal{X}}_{\mm A}}

\usepackage{booktabs}

\usepackage{framed}
\begin{document}

\maketitle

\vspace{-2.5em}

\begin{abstract}
    
    A fundamental question in the field of molecular computation is what computational tasks a biochemical system can carry out. In this work, we focus on the problem of finding the maximum likelihood estimate (MLE) for log-affine models. We revisit a construction due to Gopalkrishnan of a mass-action system with the MLE as its unique positive steady state, which is based on choosing a basis for the kernel of the design matrix of the model. We extend this construction to allow for any finite spanning set of the kernel, and explore how the choice of spanning set influences the dynamics of the resulting network, including the existence of boundary steady states, the deficiency of the network, and the rate of convergence. In particular, we prove that using a Markov basis as the spanning set guarantees global stability of the MLE steady state.
\end{abstract}

\section{Introduction}
\label{sec:intro}
A recent line of research in synthetic and systems biology is \textit{molecular computation}, which studies the problem of designing biological systems for carrying out given computations, including arithmetic~\cite{AndersonJoshi2024}, evaluation of algebraic functions~\cite{BuismanEikelderHilbersLiekens} or piecewise functions like $\min$ and $\max$~\cite{ChenDotyReevesSoloveichik2023}, approximating statistical distributions~\cite{CappellettiOrtizAndersonWinfree2020}, parameter inference \cite{Gopalkrishnan2016,VirinchiBeheraGopalkrishnan2017,ViswaBeheraGopalkrishnan2018,WiufBeheraSinghGopalkrishnan2023}, and even implementing simple neural networks~\cite{AndersonJoshiDeshpande2021}. A common approach is to construct a chemical reaction network obeying mass-action kinetics, for which the steady state is the solution to the computational problem. Realizing such reaction networks \emph{in vivo} is becoming more feasible with emerging technology, for example by DNA strand replacement techniques~\cite{SoloveichikSeeligWinfree2010,QianSoloveichikWinfree2011,ChenDalchauSrinivasPhillipsCardelliSoloveichikSeelig2013,Srinivas2015thesis,LvLiShiFanWang2021}.

In this work, we focus on the statistical problem of maximum likelihood estimation, for models that are \emph{log-affine}. This includes many important models, for instance graphical models (see \cite{lauritzen1996graphical}) and hierarchical models in contingency tables \cite{hara2012hierarchical}. A log-affine model with $m$ states and $d$ parameters is given by a design matrix $\mm{A}=(\vv a_1,\ldots,\vv a_m)\in\zz^{d\times m}$ and a scaling vector $\vv c\in\rrpp^m$, such that the set of possible probability distributions takes the form 
$$\mc M_{\mm A,\vv c}  = 
    \left\{ 
    (c_1 \vv t^{\vv a_1}, \dots, c_m \vv t^{\vv a_m})
    : 
    \vv t \in \rrpp^d, \,
    c_1 \vv t^{\vv a_1}+ \dots+ c_m \vv t^{\vv a_m}=1 
    \right\},
$$
where we use the notation $\vv t^{\vv b} = t_1^{b_1}\cdots t_d^{b_d}$ for vectors $\vv t=(t_1,\ldots,t_d)$ and $\vv b=(b_1,\ldots,b_d)$. 

The problem of maximum likelihood estimation by a reaction network was first approached by Gopalkrishnan \cite{Gopalkrishnan2016}. Specifically, it gives a mass-action system whose unique positive steady state is the maximum likelihood estimate (MLE), when the initial concentrations are given by the normalized observed distribution. A crucial part of this construction is a choice of $\zz$-basis $\Lambda$ for the right integer kernel $\ker_{\zz}(\mm{A})$. 
In this paper, we extend the construction to allow for any finite (and possibly linearly dependent) set $\Lambda\subseteq\ker_{\zz}(\mm{A})$ that spans $\ker_{\mathbb{R}}(\mm{A})$ as a vector space, and show how such $\Lambda$ can give rise to more well-behaved dynamics.

One of the main questions we address concerns the global stability of the MLE steady state. We prove that the mass-action systems of our construction are always \emph{detailed-balanced}. Such systems are known to be asymptotically stable, and we show in \Cref{thm:markov_basis_implies_no_relevant_boundary_steady_states} that the MLE is guaranteed to be a global attractor if the spanning set $\Lambda$ in the construction is sufficiently large; more precisely, if $\Lambda$ is a \emph{Markov basis}, which is a key concept in algebraic statistics (see, e.g., \cite{Almendrahernández2024} for a review). The key step in the proof is to rule out relevant boundary steady states, and the importance of the Markov basis assumption for this is highlighted in \Cref{ex:boundary_steady_states}. With this, we correct an inaccuracy in Lemma~2 and Theorem~5 of \cite{Gopalkrishnan2016}, which claimed that choosing an integer basis for $\ker_{\zz}(\mm{A})$ suffices for ruling out relevant boundary steady states.

A key geometric feature of log-affine models that is used in the construction and our global stability proof is that $\mc{M}_{\mm A,\vv c}$ can be seen as the intersection of a toric variety with the probability simplex. Toric varieties are central objects in applied algebraic geometry and have a rich combinatorial structure (see, e.g., \cite[Section~8.3]{MichalekSturmfels2021} for an introduction). A fundamental property of toric varieties that we are using, is that they can be realized as zero sets of finite collections of binomials, and this is precisely what Markov bases encode. 

The use of toric varieties in algebraic statistics has a long history; for instance, they have been used for sampling \cite{DiaconisSturmfels1998}, decomposability of graphical models \cite{geiger2006toric}, and maximum-likelihood estimation \cite{FienbergRinaldo2012, AmendolaBlissBurkeGibbonsEtAl2019,AKK2020,DDPS2023}; see also \cite{drton2008lectures} for an overview of other applications. On the reaction networks side, the use of toric geometry dates back to a series of works by Gatermann and coauthors \cite{GatermannHuber2002,GatermannEiswirthSensse2005,GatermannWolfrum2005}, and has more recently centered around the concept of toric dynamical systems \cite{CDSS09,BCS2022,CraciunJinYu2020,CJS2025} (see also the overview in \cite[Section~2]{FS2025}), conditions for monomial parametrizability (see, e.g., \cite{MillanDickensteinShiuConradi2012,MullerRegensburger2012,MillanDickenstein2018,ConradiIosifKahle2019,FH24}), and reconstruction of networks from single-cell data~\cite{WangLinSontagSorger2019}. The many toric analogies between reaction network theory and algebraic statistics have previously been discussed in \cite{CDSS09}, and will also be highlighted in the present paper.

Our paper is structured as follows. 
We recall some of the key relevant concepts from algebraic statistics and reaction network theory in \Cref{sec:alg-stat,sec:RN-intro} respectively, in order to fix notation and make the paper accessible to a broader audience. 
In \Cref{sec:main}, we present our construction of a reaction network based on a spanning set for the kernel of the design matrix, and prove that its unique positive steady state is the MLE of the corresponding log-affine model and data. \Cref{sec:choosing_spanning_set} explores the impact the choice of spanning set has on the properties of the system, including global stability, rate of convergence, and the deficiency of the reaction network. We end with a discussion and some outlooks in \Cref{sec:discussion}.

\subsection*{Notation}
In what follows, $\rrp$ denotes the set of nonnegative reals, $\rrpp$ the set of positive real numbers, and similarly $\zzp$, $\zzpp$ denote the sets of nonnegative and positive integers, respectively. We also extend the above notations to vectors,  and let $\rrp^m$ denote the nonnegative orthant and $\rrpp^m$ the positive orthant. We write $\partial\rrp^m=\rrp^m\setminus\rrpp^m$ for the boundary of the nonnegative orthant. For $n\in\zzpp$, we use the notation $[n]=\{1,\ldots,n\}$.
For a vector $\xx \in \rr^m$, we write $\xx^+$ for its nonnegative part, and $\xx^-$ for the nonpositive part; in particular, $\xx = \xx^+-\xx^-$.
For a matrix $\mm{A}$, we write $\ker_\zz(\mm{A})$ to denote the right kernel over the integers, and $\ker(\mm{A})$ to denote the right kernel over the real numbers. 
We will also make extensive use of several componentwise operations:
\begin{itemize}\itemsep -3pt
    \item multiplication: $\xx \circ \yy = (x_1y_1,x_2y_2, \ldots, x_my_m)^\top$ where $\xx \in \rr^m$, $\yy \in \rr^m$, 
    \item exponentiation: $\xx^{\vv y} = x_1^{y_1}x_2^{y_2}\cdots x_m^{y_m}$ where $\xx \in \rrp^m$, $\yy \in \rr^m$, and by convention $0^0 = 1$, 
    \item  ${\vv t}^{\mm A} = ({\vv t}^{\vv a_1}, {\vv t}^{\vv a_2}, \ldots, {\vv t}^{\vv a_m})^\top$, where ${\vv t} \in \rrp^d$ and $\mm A = \begin{pmatrix} \vv a_1,\vv a_2,\ldots, \vv a_m\end{pmatrix} \in \rr^{d \times m}$, 
    \item $\exp(\xx) = (e^{x_1},e^{x_2},\ldots, e^{x_m})^\top$ where $\xx \in \rr^m$, 
    \item $\log(\xx) = (\log x_1,\log x_2,\ldots, \log x_m)^\top$ where $\xx \in \rrpp^m$. 
\end{itemize}

\section{Algebraic statistics: log-affine models and Birch's Theorem }
\label{sec:alg-stat}

Algebraic statistics uses tools from combinatorics, commutative algebra, invariant theory, and algebraic geometry to study problems arising from probability and statistics. In this section we review important concepts from algebraic statistics like log-affine models, Markov bases, and Birch's Theorem. For a textbook reference, see~\cite[Chapters~6--7]{Sul18}.

A probability distribution on a finite state space of size $m \in \zzpp$ is determined by probabilities  $p_j \geq 0$ for each $j=1,\dots,m$, and can hence be seen as an element of the $(m-1)$-dimensional \demph{probability simplex}
\eq{ 
    \Delta_{m-1} \coloneqq \left\{ (p_1,p_2,\ldots, p_m) \in \rrp^m : \sum_{i=1}^m p_i = 1  \right\}. 
}
A subset of the probability simplex is called an \demph{algebraic statistical model} when it is parameterized by a polynomial map or when an implicit description is given by polynomials.
The focus of this paper is a particularly well-studied class of algebraic statistical models known as \emph{log-affine models}; see, e.g.,  \cite{haber1985maximum,fienberg2007three,FienbergRinaldo2012}, and \cite[Section~6.2]{Sul18}.

\begin{defn}
\label{def:log-affine-model}
    Given a \demph{design matrix}  $\mm A = \begin{pmatrix} \vv a_1, \vv a_2,\ldots, \vv a_m\end{pmatrix} \in \zz^{d \times m}$ that is assumed to contain the vector $\mathbbm{1} \coloneqq (1, 1, \dots, 1)$ in its rowspan, and $\vv c \in \rrpp^m$, the associated \demph{log-affine model}  is 
    \eq{ 
        \mc M_{\mm A,\vv c}  \coloneqq  \{ \vv{p} \in \Int(\Delta_{m-1}) : \log \vv{p} \in \log \vv c + \row({\mm A}) \} . 
    }
    In the case when $\vv c = \mathbbm{1}$, the model $\mc M_{\mm A, \mathbbm{1}}$, denoted $\mc M_{\mm A}$, is said to be \demph{log-linear}. 
\end{defn}

To see that a log-affine model is an algebraic statistical model, consider the following description of $\mc M_{\mm A,\vv c}$ via the monomial map
\begin{equation}\label{eq:parametrization_of_log_affine_set}
    \phi_{\mm{A},\vv{c}}\colon\rrpp^d\to\rrpp^m,\quad \vv{t}\mapsto \vv{c}\circ\vv{t}^{\mm A}=(c_1 \vv t^{\vv a_1}, \dots, c_m \vv t^{\vv a_m})\,.
\end{equation}
Its image is a scaled positive toric variety, which we denote as $\mc X_{\mm A,\vv c}$ (or $\Xa$ when $\vv{c}=\mathbbm{1}$), i.e.,
\begin{equation}\label{eq:define-xac}
    \Xac =\left\{ 
        (c_1 \vv t^{\vv a_1}, \dots, c_m \vv t^{\vv a_m})
        : \vv t \in \rrpp^d 
        \right\}\subseteq\rrpp^m.
\end{equation}
The model $\mc{M}_{\mm A,\vv c}$ is obtained by restricting $\Xac$ to the probability simplex:
\[
    \mc M_{\mm A,\vv c}  = \Xac \cap\Delta_{m-1}=
    \left\{ 
        (c_1 \vv t^{\vv a_1}, \dots, c_m \vv t^{\vv a_m})
        : 
        \vv t \in \rrpp^d, \,
        c_1 \vv t^{\vv a_1}+ \dots+ c_m \vv t^{\vv a_m}=1
    \right\} . 
\]
After logarithmizing, $\Xac$ is an affine set. From this, we easily obtain an implicit description in terms of binomials; for any set $\Lambda\subseteq\ker_{\zz}(\mm{A})$ that spans $\ker(\mm{A})$ as a vector space, it holds that
\begin{equation}\label{eq:xac-implicit}        
         \Xac = 
            \left\{ \xx \in \rrpp^m :  \vv{c}^{\vv{\gamma}^+}\xx^{\vv{\gamma}^-}=\vv{c}^{\vv{\gamma}^-}\xx^{\vv{\gamma}^+}\text{for all}\:\: \vv \gamma\in\Lambda \right\}. 
\end{equation}

Since we want to treat distributions where some entries are zero, we also consider the Euclidean closure of $\Xac$ in $\rrp^m$:
\[
    \closedXac\coloneqq\overline{\left\{ 
        (c_1 \vv t^{\vv a_1}, \dots, c_m \vv t^{\vv a_m})
        : \vv t \in \rrpp^d 
        \right\} }\subseteq\rrp^m.
\]
Obtaining an implicit description of $\closedXac$ is a more subtle task than for $\Xac$, since we cannot logarithmize vectors with zeroes. In particular, the binomials given by a spanning set $\Lambda$ as in \eqref{eq:xac-implicit}  might cut out too large a subset of $\rrp^m$ (see \Cref{ex:log_linear_model} below).
Nevertheless, the Zariski closure of $\mc X_{\mm A,\vv c}$ in $\cc^m$ is a scaled toric variety, so it follows from the theory of toric ideals that there is an implicit description in terms of binomials, given by what is called a Markov basis.

\begin{defn}
\label{def:Markov} 
A \demph{Markov basis} for a design matrix $\mm{A}\in\zz^{d\times m}$ 
is a finite set 
$\Lambda\subseteq\ker_{\zz}(\mm{A})$ such that the binomials 
$\vv{x}^{\vv{\gamma}^+}-\vv{x}^{\vv{\gamma}^-}$ 
for $\vv{\gamma}\in\Lambda$ generate the ideal 
$I(\Xa)\subseteq\rr[x_1,\ldots,x_m]$ 
of all polynomials that vanish on $\Xa$. 
\end{defn}

That such a finite generating set exists for any design matrix follows from Hilbert's Basis Theorem and \cite[Proposition~1.1.9]{cox2011toric}. The concept of Markov bases has a long history in algebraic statistics \cite{Almendrahernández2024}, and there are many software packages available for computing a Markov basis for a given design matrix, for instance \texttt{4ti2} \cite{4ti2}. In what follows, we will need the following well-known lemma, which follows from the definition of a Markov basis and elementary facts about lattice ideals (see, e.g., \cite[Proposition 7.5]{miller2005combinatorial}).


\begin{lem}\label{lem:markov_basis_cuts_out_closedXac}
Let $\Lambda\subseteq\ker_{\zz}(\mm{A})$ be a Markov basis for $\mm{A}\in\zz^{d\times m}$. Then $\Lambda$ is a spanning set of $\ker(\mm{A})$, and for any $\vv{c}\in\rrpp^m$ it holds that
\begin{equation*}
    \closedXac =   \left\{ \xx \in \rrp^m :  \vv{c}^{\vv{\gamma}^+}\xx^{\vv{\gamma}^-}=\vv{c}^{\vv{\gamma}^-}\xx^{\vv{\gamma}^+} \text{for all}\:\: \vv \gamma\in\Lambda \right\}.
\end{equation*}    
\end{lem}

\begin{ex}\label{ex:log_linear_model}
The following will appear throughout the paper as a running example.
Here, we use it to illustrate the importance of picking a sufficiently large spanning set $\Lambda$ when considering the Euclidean closure  $\closedXac$.
Let $d=2$ and $m=4$, and consider the log-linear model given by
\[\mm A=\begin{pmatrix}4&2&3&1\\ 0&2&1&3\end{pmatrix}.\]
This defines the set 
\begin{equation}\label{eq:log_linear_model-Xac}
\Xa = \left\{(t_1^4,t_1^2t_2^2,t_1^3t_2,t_1t_2^3):(t_1,t_2)\in\rrpp^2\right\}.
\end{equation}
An example of a Markov basis for $\mm{A}$ is 
$$\Lambda_\mathrm{mb}=\{(1,1,-2,0),\:(0,2,-1,-1),\: (2,0,-3,1),\: (1,-1,-1,1)\},$$ 
and it follows from \Cref{lem:markov_basis_cuts_out_closedXac} that
\[\closedXa = \{(x_1,x_2,x_3,x_4)\in\rrp^4 : x_1x_2=x_3^2,\: 
 x_2^2=x_3x_4,\:  x_1^2x_4=x_3^3,\:  x_1x_4=x_2x_3 \}.\]
On the other hand, the vector space basis
$$\Lambda_1=\{(1,1,-2,0), (0,2,-1,-1)\}$$ 
is not a Markov basis, and the resulting binomials cut out a strictly larger set in $\rrp^4$:
\[\{(x_1,x_2,x_3,x_4)\in\rrp^4:  x_1x_2=x_3^2 , \:x_2^2=x_3x_4 \}=\closedXa\cup\{(0,x_2,x_3,0) : x_2,x_3\in\rrp\}.\]
In contrast, it follows from \eqref{eq:xac-implicit} that the restriction to the \emph{positive} orthant $\rrpp^4$ can be described by either spanning set:
\begin{align*}
    \Xa &= \{ (x_1,x_2,x_3,x_4)\in\rrpp^4 : x_1x_2=x_3^2,\: 
 x_2^2=x_3x_4,\:  x_1^2x_4=x_3^3,\:x_1x_4 = x_2x_3 
 \} \\
    &= \{(x_1,x_2,x_3,x_4)\in\rrpp^4 :  x_1x_2=x_3^2,\: 
 x_2^2=x_3x_4\}.
\end{align*}
\end{ex}

Thus far, we have focused on parametric and implicit descriptions of log-affine models.  Now we turn to \emph{maximum likelihood estimation}.
Suppose we observe independently and identically distributed data $\vv{u}=(u_1,\ldots,u_m)\in\rrp^m\setminus\{\vv{0}\}$ where $u_j$ is the number of observations of state $j$. We denote the total number of observations by $u_+=\sum_{j=1}^m u_j$, and the corresponding observed distribution by $\bar{\vv{u}}=\vv{u}/u_+\in\Delta_{m-1}$. A
\demph{maximum likelihood estimate (MLE)} of $\vv{u}$ for a model $\mathcal{M}\subseteq\Delta_{m-1}$ is a distribution  $\widehat{\vv{p}}\in\mathcal{M}$ that ``best explains'' the data, in the sense that it maximizes the likelihood function
\begin{equation*}
    L_{\vv u} \colon \mathcal{M}\to \mathbb{R},\quad
    \vv p\mapsto p_1^{u_1} \dots p_m^{u_m}. 
\end{equation*}

For log-linear models, there is a well-known characterization of the (unique) MLE as the point in the model that satisfies some linear constraints. This will later be the foundation for the construction in \Cref{sec:main}. In the algebraic statistics literature, this and similar results are often referred to as Birch's Theorem \cite[Section~1.2]{PachterSturmfels2005}. For proofs of the version we give here, see \cite[Proposition~2.1.5]{drton2008lectures}, \cite[Theorem 8.2.1]{Sul18}, and \cite[Theorem~4.8]{lauritzen1996graphical}.

\begin{thm}[Birch's Theorem]\label{thm:Birch}
Let $\mm A \in \mathbb{Z}^{d \times m}$ with $\mathbbm{1}\in\row(\mm A)$, let $\vv c \in \rrpp^m$, and let $\vv{u}\in\rrp^m$ be a vector such that $\bar{\vv{u}}+\ker(\mm{A})$ intersects $\rrpp^m$. Then there is a unique solution to the system 
\begin{equation*}
    \widehat{\vv{p}} \in \closedXac
    \quad \text{and}\quad 
    {\mm A} \widehat{\vv{p}}={\mm A} \bar{\vv{u}}.
\end{equation*}
This solution $\widehat{\vv{p}}$ has positive entries, and constitutes the MLE of $\vv{u}$ for the model $\mathcal{M}_{\mm A,\vv c}$.
\end{thm}

From the discussion above, it is clear that the MLE is an algebraic function of the observed data. Interestingly, it is sometimes a rational function of $\vv{u}$, which the following example illustrates. Determining which log-affine models have this property and how the rational MLE depends on $\mm{A}$ and $\vv{c}$ is a central topic in algebraic statistics, and has been studied in, e.g., \cite{huh2014-mldegree-one, DMS2021-rational-mle,DDPS2023}.

\begin{ex}\label{ex:hardy_weinberg}
An additional example, which we will use for figures and explicit calculations throughout the paper, is the \demph{Hardy--Weinberg model}, given by
\[\mm{A}=\begin{pmatrix}0 & 1 & 2\\2 & 1 & 0\end{pmatrix}\quad\text{and}\quad\vv{c}=(1,2,1).\]
The name alludes to the fact that this model appears in population genetics: each point $(t_1^2,2t_1t_2,t_2^2)\in\mathcal{M}_{\mm A,\vv c}$ encodes the distribution of the three possible diploid genotypes for a gene with two alleles with frequencies $t_1$ and $t_2$. The MLE is a rational function of the observed data $\vv{u}\in\zzpp^m$ (see, e.g., \cite[Example 1.3]{Likelihood-geometry-survey}):
\begin{equation}
\label{eq:hardy_weinberg_mle}
\widehat{\vv{p}}=   \left( \frac{(2u_1+u_2)^2}{4(u_1+u_2+u_3)^2} , \frac{(u_2+2u_3)(2u_1+u_2)}{2(u_1+u_2+u_3)^2}  , \frac{(u_2+2u_3)^2}{4(u_1+u_2+u_3)^2} \right).\end{equation}
The model, the positive toric variety, and the MLE are illustrated in \Cref{fig:ex-ALG-stoichiometric_compatibility_class}.
\end{ex}

\begin{figure}[hbt!]
\centering
    \includegraphics[width=2.5in]{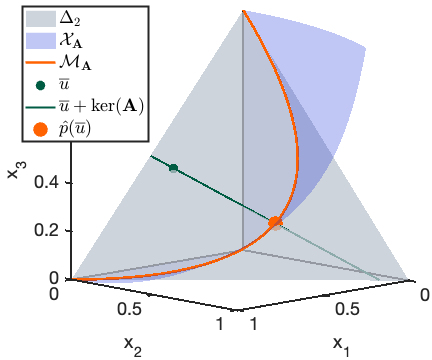}
    \caption{
    Illustration of Birch's Theorem for the Hardy--Weinberg model (\Cref{ex:hardy_weinberg}). The model $\mc{M}_{\mm A}$ (orange) is the intersection of $\Xac$ (blue) and the probability simplex $\Delta_2$ (gray).
    Birch's Theorem guarantees that for any  observed distribution $\bar{\vv u}\in\Delta_{2}$, there is a unique point $\widehat{\vv p}$ in the intersection between $\closedXa$ and $\bar{\vv u} + \ker(\mm A)$ (green), and that this is the MLE.}
    \label{fig:ex-ALG-stoichiometric_compatibility_class}
\end{figure}

We end the section by introducing two more examples of  models that play an important role in the algebraic statistics literature, and which we will use for our numerical experiments  in \Cref{sec:slow-mfld}. 
For the first model, the MLE depends rationally on the observed data $\vv{u}$, whereas this is not the case for the other model.

{\samepage
\begin{ex}\label{ex:independence}
Given two random variables with $r_1$ and $r_2$ states, respectively, their joint distribution is a point $\vv{p}=(p_{ij})_{(i,j)\in[r_1]\times [r_2]}\in\Delta_{r_1r_2-1}$, where $p_{ij}$ is the probability that the first random variable is in state $i$ and the second is in state $j$. If the two random variables are independent, the joint distribution belongs to the \demph{independence model}, which is the image of the parametrization
\begin{equation*}
    \Delta_{r_1-1}\times \Delta_{r_2-1}\to \Delta_{r_1r_2-1}, \quad ((a_1,\dots, a_{r_1}), (b_1,\dots, b_{r_2}))\mapsto (a_ib_j)_{(i,j)\in [r_1]\times [r_2]}.
\end{equation*}
This is a log-linear model, corresponding to the design matrix
\[
\mm{A}_\mathrm{ind}(r_1,r_2)\coloneqq
\left(\begin{array}{ccc ccc c ccc}
1 & \cdots & 1 
\\
  &        &   & 1 & \cdots & 1 
  \\
  &        &   &   &        &   & \ddots 
  \\
  &        &   &   &        &   &        & 1 & \cdots & 1 \\[0.5em]\hline\\[-0.5em]
1 &      &  & 1 &       &  & & 1 &       &  \\
  & \ddots &   &   & \ddots &   & \cdots &   & \ddots &   \\
  &        & 1 &   &        & 1 &  &   &        & 1 \\
\end{array}\right)\in\zz^{(r_1+r_2)\times r_1r_2}.
\]
This model has a unique minimal Markov basis \cite[Proposition 1.2.2]{drton2008lectures} given by
\begin{equation}\label{eq:markov_indep}
  \pm( \vv{e}_{ij}+\vv{e}_{k\ell}-\vv{e}_{i\ell}-\vv{e}_{kj}) , \quad \text{ for }\: 1 \leq i < k \leq r_1,\:\: 1 \leq j < \ell \leq r_2
\end{equation}
where $\vv{e}_{ij}$ is the standard basis vector of $\rr^{r_1r_2}$ corresponding to the $p_{ij}$ coordinate. 
It can be verified that the following rational expression in $\vv u$ satisfies the conditions of Birch's Theorem, and therefore is the MLE:
\begin{equation}\label{eq:MLE_independence_model}
\widehat{p}_{ij}=\frac{\big(\sum_{\ell=1}^{r_2} u_{i\ell}\big)\big(\sum_{k=1}^{r_1} u_{kj}\big)}{\big(\sum_{k=1}^{r_1}\sum_{\ell=1}^{r_2} u_{k\ell}\big)^2}.
\end{equation}
\end{ex}
}

\begin{ex}\label{ex:4cycle}
The theory of graphical models is a rich source of examples of log-linear models. A particularly interesting example is the \demph{undirected 4-cycle} \cite[Example~4]{geiger2006toric}, which has a $16\times 16$ design matrix, and is the smallest graphical model for which the MLE does not depend rationally on the observed distribution.
\end{ex}

\section{Reaction network theory:  detailed-balanced systems
}
\label{sec:RN-intro}
Reaction network theory is the study of certain families of parametrized systems of ordinary differential equations (ODEs) encoded by directed graphs. In this section, we review some basic notions in reaction network theory, including the notion of \emph{detailed-balanced} steady states. For a more comprehensive introduction to the field, see the review articles \cite{Gunawardena_review, Dic16,YuCraciun2018,FS2025}, and the textbook \cite{Fei19}.

\begin{defn}\label{defn:reaction-network}
    A \demph{reaction network} with species $\{\cf{X}_1,\ldots,\cf{X}_m\}$ is a simple directed graph $G = (V,E)$ with no isolated vertices, where the vertices are points in $\zzp^m$ (interpreted as formal linear combinations of the species), referred to as \demph{complexes}, and the edges, denoted $\yy_i \to \yy_j$ for $\yy_i,\yy_j\in V$, are referred to as \demph{reactions} of the network. 
    A reaction network  $G$ is said to be \demph{reversible} if every reaction is part of a reversible pair, i.e., $\yy_i \to \yy_j \in E$ if and only if $\yy_j \to \yy_i \in E$.   We write $\yy_i \RR \yy_j$ for such a reaction pair. 
\end{defn}

The goal of reaction network theory is to model the time evolution of the concentrations $\xx=(x_1,\ldots,x_m)^\top$ of the species $\cf{X}_1,\ldots,\cf{X}_m$ of the network under some kinetic assumptions. One of the most common, \emph{mass-action kinetics}, is based on the first principle that the reaction rate is proportional to the collision probability of the reacting species.

\begin{defn}
    A \demph{mass-action system} is a weighted directed graph $(G,\vv\kk)$, where $G$ is a reaction network, and $\vv \kk = (\kk_{ij})_{\vv{y}_i\to\vv{y}_j\in E}$ is a tuple of positive real numbers called \demph{rate constants}, indexed by the reactions. The induced dynamics of a mass-action system $(G,\vv{k})$ is given by the ODEs
    \eqn{\label{eq:mas-ode}
        \frac{d\xx}{dt} = \sum_{\vv{y}_i\to\vv{y}_j\in E} \kk_{ij} \xx^{\yy_i} (\yy_j - \yy_i),
    }
    where the right-hand side is a vector of polynomials that we will denote $\vv{f}_{G,\vv{k}}(\vv{x})$.
    For an initial value $\vv{x}_0\in\rrp^m$, we denote the corresponding trajectory of \eqref{eq:mas-ode} by $\vv{x}(t;\vv{x}_0)$ or simply $\vv{x}(t)$.
\end{defn}
 
The nonnegative and positive orthants are \emph{forward-invariant} with respect to \eqref{eq:mas-ode} \cite{Volpert1973,Sontag2001}. In other words, if $\vv{x}_0\in\rrp^m$, then $\xx(t;\vv{x}_0)\in\rrp^m$ for all $t\geq 0$ for which the trajectory is defined, and, analogously, if $\vv{x}_0\in\rrpp^m$, then $\xx(t;\vv{x}_0)\in\rrpp^m$ for all $t\geq 0$ for which the trajectory is defined.

The \demph{stoichiometric subspace} of a reaction network $G$ is the linear subspace 
\[
    S \coloneqq \spn_\rr \{ \yy_j - \yy_i  : \yy_i \to \yy_j \in E\}\subseteq\rr^m.
\]  
Each initial value $\xx_0$ defines a \demph{stoichiometric compatibility class}
$$(\xx_0+S)_{\geq 0} \coloneqq \{\xx_0+\vv s : \vv s\in S \}\cap \rrp^m,$$ 
where the part intersecting the positive orthant will be denoted  $(\xx_0+S)_{>0}$.
Since the right-hand side of \eqref{eq:mas-ode} lies in $S$ for any $\xx \in \rrpp^m$, the subset  $(\xx_0+S)_{>0}$ is also forward-invariant. 
If $\dim(S) < m$, the nontrivial stoichiometric compatibility classes are given by \demph{conservation laws}, which are linear first integrals, represented by a linear system $\mm W \xx = \mm W \xx_0$, where the rows of $\mm W$ span $S^\perp$, and $\xx_0$ is the initial state. The linear forms $\mm W \xx$ are referred to as \demph{conserved quantities}, and $\mm W \xx_0$ as \demph{total amounts}.

\begin{ex}
\label{ex:network} 
    Consider a reaction network consisting of two reversible reactions: 
    \begin{equation}
    \label{crn:running_ex-Lambda1}
    \vcenter{\hbox{\begin{tikzpicture}
            \begin{scope}[shift={(0,0)}]
            \node (1a) at (0,0) [left]  {$2\,\cf{X}_3$} ; 
            \node (1b) at (1.25,0) [right] {$\cf{X}_1 + \cf{X}_2$,}; 
            \draw[fwdrxn, transform canvas={yshift=1.75pt}] (1a)--(1b) node [midway, above] {\ratecnst{$\kk_{12}$}};
            \draw[fwdrxn, transform canvas={yshift=-1.75pt}] (1b)--(1a)  node [midway, below] {\ratecnst{$\kk_{21}$}};
            \end{scope}
            \begin{scope}[shift={(5,0)}]
        \node (2a) at (0,0) [left] {$\cf{X}_3 + \cf{X}_4$}; 
        \node (2b) at (1.25,0) [right] {$2\,\cf{X}_2$.};
        \draw[fwdrxn, transform canvas={yshift=1.75pt}] (2a)--(2b) node [midway, above] {\ratecnst{$\kk_{34}$}};
        \draw[fwdrxn, transform canvas={yshift=-1.75pt}] (2b)--(2a) node [midway, below] {\ratecnst{$\kk_{43}$}};
    \end{scope}
        \end{tikzpicture}}}
    \end{equation}
    Its associated ODE system is
    \begin{align*}
        \frac{d\vv{x}}{dt}
        = \begin{pmatrix*}[l]
            k_{12}x_3^2 - k_{21} x_1x_2\\
             k_{12}x_3^2 - k_{21} x_1x_2 + 2k_{34}x_3x_4 - 2k_{43}x_2^2 \\
             -2k_{12}x_3^2 + 2k_{21} x_1x_2 - k_{34}x_3x_4 + k_{43}x_2^2 \\
            -k_{34}x_3x_4   + k_{43} x_2^2
        \end{pmatrix*}. 
    \end{align*}
    Its stoichiometric subspace $S$ is the two-dimensional subspace spanned by $(1,1,-2,0)^\top$ and $(0,2,-1,-1)^\top$. 
    Two conservation laws are expected since $\codim(S) = 2$; for example, 
    \eq{ 
        \begin{pmatrix} 
            4 & 2 & 3 & 1 \\ 
            0 & 2 & 1 & 3 
        \end{pmatrix} ( \vv x - \vv x_0) 
        = 
        \vv{0},  
    }  
    as the rows of the matrix span $S^\perp$. 
    One can also confirm the conservation laws by checking that $4\dot x_1 + 2\dot x_2 + 3 \dot x_3 + \dot x_4= 0$ and $2 \dot x_2 + \dot x_3 + 3\dot x_4 = 0$. The choice of conservation laws is not unique; for example one can also check that $\dot x_1 + \dot x_2 + \dot x_3 + \dot x_4 = 0$.  
\end{ex}

Mass-action systems are capable of very diverse dynamics; they can display oscillation, have multiple stable steady states in the same stoichiometric compatibility class, and even exhibit chaotic dynamics (e.g., \cite[Figure~21]{GasparToth2023}).  That said, some classes of mass-action systems are known to be particularly stable. The most well-known are \emph{detailed-balanced systems}, which originated from Boltzmann~\cite{Boltzmann_1887, Boltzmann_1896}. 

\begin{defn}\label{def:detailed_balanced}
Let $(G,\vv\kk)$ be a mass-action system in $\rr^m$ with associated ODEs $\dot{\xx} = \vv f_{G,\vv\kk}(\xx)$.
\begin{enumerate}[label=(\alph*)]
    \item A point $\widehat{\xx} \in \rrp^m$ is a \demph{steady state} if $\vv f_{G,\vv\kk}(\widehat{\xx}) = \vv 0$. It is a \demph{positive steady state} if $\widehat{\xx} \in \rrpp^m$; otherwise it is a \demph{boundary steady state}. The set of all positive steady states is denoted $\mc{E}_{G,\vv\kk}$.
    
    \item If the network $G$ is reversible, then a positive steady state  $\widehat{\xx}$ is \demph{detailed-balanced} if       
    \begin{equation}\label{eq:detailed_balanced}
            \kk_{ij}\,\widehat{\xx}^{\yy_i}=\kk_{ji}\, \widehat{\xx}^{\yy_j}
            \text{ for all }
            \vv{y}_i\RR\vv{y}_j.   
        \end{equation}
    The set of positive detailed-balanced steady states is denoted $\mc{D}_{G, \vv \kk}$.
\end{enumerate}
\end{defn}

If a mass-action system $(G,\vv{k})$ has a positive steady state that is detailed-balanced, then 
all positive steady states are detailed-balanced; therefore, we say $(G,\vv\kk)$ is a \demph{detailed-balanced mass-action system} 
if $\mc{D}_{G,\vv\kk}$ is nonempty.
This, and other key properties of detailed-balanced systems, are collected in the following theorem. 
We note that these results were originally proved for the more general class of \emph{complex-balanced} mass-action systems (see \cite{HJ72} for a definition and precise statements).

\begin{thm}[\cite{HJ72}]
\label{thm:DB_properties}
Let $(G,\vv\kk)$ be a reversible mass-action system with stoichiometric subspace $S$, and suppose there exists a detailed-balanced steady state $\widehat{\xx} \in \mc{D}_{G,\vv\kk}$. Then the following holds:
\begin{enumerate}[label={\normalfont(\alph*)}]
\setlength{\itemsep}{0.1em}
    \item\label{thm-it:DB-syst}  $\mc{D}_{G,\vv\kk} = \mc{E}_{G,\vv\kk}$. 
    
    \item\label{thm-it:DB-parametrization} The set of positive steady states is given by
    \[\mc{E}_{G,\vv\kk} = \{ \xx \in \rrpp^m : \log \xx - \log \widehat{\xx} \in S^\perp\}=\{\widehat{\xx} \circ \vv t^{\mm W}:  \vv t\in\rrpp^{d}\}\] for any matrix $\mm W\in\rr^{d\times m}$ whose rows span $S^\perp$.
    
    \item\label{thm-it:DB-Birch}   There is exactly one steady state in every positive stoichiometric compatibility class $(\xx_0 + S)_{>0}$.

    \item\label{thm-it:DB-stability}
    The unique positive steady state in $(\vv{x}_0+S)_{>0 }$ is asymptotically stable within $(\xx_0 + S)_{>0}$.
    
\end{enumerate}
\end{thm}

In the language of \Cref{sec:alg-stat}, part \ref{thm-it:DB-parametrization} says that $\mc{E}_{G,\vv k}=\mc{X}_{\mm{W},\widehat{\xx}}$, and part \ref{thm-it:DB-Birch} corresponds to Birch's Theorem (\Cref{thm:Birch}), with the observed distribution playing the role of $\xx_0$ and the MLE playing the role of the detailed-balanced steady state $\widehat{\xx}$. 

\begin{rmk}
   The strictly convex function $V(\xx) = \xx {}\cdot (\log \xx - \log \xx^* - \mathbbm{1})$, known as the free energy or entropy function in thermodynamics, is a Lyapunov function defined on $\rrpp^m$ for the system \eqref{eq:mas-ode}, with global minimum at $\xx^*$. The similarity to the entropy of a discrete distribution foreshadows the correspondence between positive steady states and MLEs spelled out in \Cref{sec:main}; while the former minimize $V$, the latter minimize the \emph{Kullback--Leibler divergence} $\operatorname{KL}(\bar{\vv u} || \vv p)$ to the empirical distribution $\bar{\vv u}$. 
   See \cite[Section~2]{invariant-toric} and  \cite[Proposition~11]{CDSS09}.
\end{rmk}

\begin{ex}
\label{ex-DYN-stoichiometric_compatibility_class-a}
    Consider the following network with a single reversible reaction pair:
    \begin{equation*}
    \vcenter{\hbox{\begin{tikzpicture}
            \begin{scope}[shift={(0,0)}]
            \node (1) at (0,0) [left] {$\cf{X}_1 + \cf{X}_3$ }; 
            \node (2) at (1.8,0)  {$2\, \cf{X}_2$,};
            \draw[fwdrxn, transform canvas={yshift=2pt}] (1)--(2) node[above, midway] {\ratecnst{$4$}};
            \draw[fwdrxn, transform canvas={yshift=-2pt}] (2)--(1) node[below, midway] {\ratecnst{$1$}};
            \end{scope}
        \end{tikzpicture}}}
\end{equation*}
The corresponding mass-action system is
$$\dot{x}_1 = x_2^2-4x_1x_3,\quad \dot{x}_2= -2x_2^2+8x_1x_3,\quad \dot{x}_3=x_2^2-4x_1x_3,
$$
and two linearly independent conserved quantities are given by the rows of the matrix
\[\mm{W}=\begin{pmatrix}2 & 1 & 0\\ 0 & 1 &2\end{pmatrix}.\]
The system is detailed-balanced, where 
\[\mathcal{E}_{G,\vv{k}}=\mathcal{D}_{G,\vv{k}}=\{(t_1^2,2t_1t_2,t_2^2):(t_1,t_2)\in\rrpp^2\}\subseteq\rrpp^3.\]
\Cref{fig:ex-DYN-stoichiometric_compatibility_class} shows the phase portrait restricted to $\Delta_2$, a stoichiometric compatibility class with an initial state $\xx_0$ and the corresponding steady state $\hat{\xx}$. The similarity with \Cref{fig:ex-ALG-stoichiometric_compatibility_class} is not a coincidence; we will see in \Cref{sec:main} that this network is obtained by \Cref{const:DB_system} for the MLE problem of \Cref{ex:hardy_weinberg}.
\end{ex}

\begin{figure}[h!]
\centering
\includegraphics[width=2.5in]{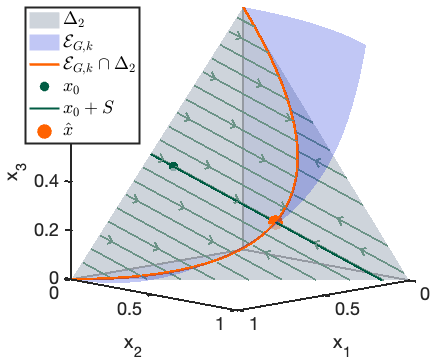}
    \caption{Illustration of \Cref{ex-DYN-stoichiometric_compatibility_class-a}, showing the set of positive steady states $\mathcal{E}_{G,\vv k}$ (blue), an initial state $\vv{x}_0$ and the corresponding stoichiometric compatibility class $\vv{x}_0+S$ (green), and the unique positive steady state $\widehat{\vv{x}}$ (orange), and the phase portrait restricted to $\Delta_2$ (light green). Note that $\dim(S) = 1$.}
    \label{fig:ex-DYN-stoichiometric_compatibility_class}
\end{figure}

An important property that we consider in the next section is \emph{global convergence}, and we therefore devote the remainder of this section to discussing global convergence properties of detailed-balanced systems.

A long-standing conjecture in reaction network theory is that the unique positive steady state in each stoichiometric compatibility class is not only locally stable, but \emph{globally stable}~\cite{Hor74}. This conjecture is referred to as the \emph{Global Attractor Conjecture}, and has a rich history (see \cite{CraciunNazarovPantea2013,Craciun2019}). It has been proven for several classes of networks, including networks where the underlying graph is strongly connected~\cite{Anderson2011, BorosHofbauer2020}; other proven cases involve graph-theoretic and geometric assumptions on the network, see, e.g., \cite{Pantea2012, CraciunNazarovPantea2013, GopalkrishnanMillerShiu2014}. 

It has been proven that if a positive steady state of a complex-balanced system fails to be a global attractor, then there exist trajectories that approach the boundary \cite{SiegelMacLean2000,Angeli2007,Shiu2010}. Below, we state this result in the special case of detailed-balanced systems. Recall that the \demph{$\omega$-limit set} of a trajectory $\xx(t; \xx_0)$ is the set of subsequential limit points, i.e., 
$$\omega(\xx_0) = \left\{\vv{y}\in\rr^n:\text{$\vv{y}=\lim_{n\to\infty}\vv{x}(t_n;\vv{x}_0)$ for some increasing sequence $(t_n)_{n=1}^\infty$}\right\}.$$

\begin{prop}[{\cite[Theorem~3.2]{SiegelMacLean2000}}]\label{prop:omega_limits}
    Let $(G,\vv\kk)$ be a detailed-balanced mass-action system, and let $\xx_0 \in \rrpp^m$.  
    Then the $\omega$-limit set $\omega(\xx_0)$ either consists of boundary steady states or the unique positive steady state in $(\xx_0 +S)_{>0}$. 
\end{prop}

A network is said to \demph{lack relevant boundary steady states} if no stoichiometric compatibility class that intersects the positive orthant contains a steady state on the boundary of the nonnegative orthant.  A well-known consequence of \Cref{prop:omega_limits} is that this property implies global convergence.

\begin{prop}[\cite{SiegelMacLean2000}]
\label{prop:no_boundary_steady_states_imply_global_convergence}
    Let $(G,\vv{k})$ be a detailed-balanced mass-action system that lacks relevant boundary steady states. Then for any $\vv{x}_0\in\rrpp^m$, the  trajectory $\vv{x}(t;\vv{x}_0)$ converges to the unique positive steady state $\widehat{\vv{x}}$ in $\mc{E}_{G,\vv{k}}\cap (\vv{x}_0+S)_{>0}$.
\end{prop}

A useful condition to rule out relevant boundary steady states comes from the theory of siphons, developed in \cite{Angeli2007} and expanded in several subsequent works, e.g., \cite{Shiu2010,Deshpande2014,AlRadhawiAngeli16,Marcondes2017}.

\begin{defn}\label{def:siphon}
A \demph{siphon} in a reaction network $G = (V,E)$ with species $\{\cf{X}_1,\ldots,\cf{X}_m\}$ is a set of indices $Z\subseteq [m]$ such that for every $i\in Z$ and every reaction $\vv{y}\to \vv{y}' \in E$ with $y'_i>0$, there exists $j\in Z$ such that $y_j>0$. 
\end{defn}

Intuitively, a siphon keeps track of which species can simultaneously have zero concentrations at equilibrium.
It turns out that there cannot exist relevant boundary steady states if \lq\lq mass is conserved\rq\rq\ within each minimal siphon, in the following sense.

\begin{prop}[{\cite{Angeli2007}, \cite[Proposition~2(v)]{Marcondes2017}}]\label{prop:siphon_critierion}
Consider a network with stoichiometric subspace $S\subseteq\rr^m$. Suppose that for every inclusion-minimal siphon $Z\subseteq [m]$, there exists $\vv w\in S^\perp\cap\rrp^m$ with $\supp(\vv w)=Z$. Then the network lacks relevant boundary steady states.
\end{prop}

A network that satisfies the condition in \Cref{prop:siphon_critierion} is said to be \demph{structurally persistent}. 

We end this section by using the notion of siphons to extend \Cref{thm:DB_properties}\ref{thm-it:DB-syst} and \Cref{prop:no_boundary_steady_states_imply_global_convergence} to also account for the dynamics on the boundary $\partial\rrp^m$, which we will need in \Cref{sec:boundary}. 

\pagebreak

Consider the set of nonnegative steady states 
\[
    \mc{E}_{G,\vv \kk}^{\geq 0} \coloneqq 
    \left\{
        \xx \in \rrp^m : 
        \vv{f}_{G,\vv\kk}(\xx) = \vv 0
    \right\}.
\]
We say $\widehat{\xx} \in \mc{E}_{G,\vv \kk}^{\geq 0}$ is a \demph{nonnegative detailed-balanced steady state} if it satisfies \eqref{eq:detailed_balanced}; we denote the set of such points by $\mc{D}_{G,\vv \kk}^{\geq 0}$. 
The following is the detailed-balanced analog of \cite[Theorem~3.1]{CappellettiWiuf2016}.

\begin{prop}\label{prop:all_nonnegative_steady_states_are_DB}
Let $(G,\vv{k})$ be a detailed balanced mass-action system. Then the set of nonnegative steady states and the set of nonnegative detailed-balanced states coincide, i.e., $\mc{E}_{G,\vv k}^{\geq 0}=\mc{D}_{G,\vv k}^{\geq 0}$. 
\end{prop}
\begin{proof}
    The inclusion $\mc{D}_{G,\vv k}^{\geq 0}\subseteq\mc{E}_{G,\vv k}^{\geq 0}$ is immediate. To prove the inclusion $\mc{E}_{G,\vv k}^{\geq 0}\subseteq\mc{D}_{G,\vv k}^{\geq 0}$, we begin by introducing some notation.  For a set $Z\subseteq[m]$, let $G^Z$ denote the subnetwork of $G$ obtained by removing all reactions where a species with index in $Z$ appears as a reactant, and let $\vv{k}^Z$ be the restriction of $\vv{k}$ to the reactions of $G^Z$. For any vector $\xx \in \rrp^m$ and $Z \subseteq [m]$, we denote by $\xx^Z$ the vector in $\rrp^m$ where $x^Z_i = x_i$ for all $i \not\in Z$, and $x^Z_i = 0$ for all $i \in Z$. Finally, by a slight abuse of notation, we view the sets of positive steady states $\mc{E}_{G^{Z},\vv k^Z}$ and positive detailed-balanced steady state $\mc{D}_{G^Z,\vv k^Z}$ of the subnetwork as being embedded in $\rr^m$, by supplementing with $\widehat{x}_i = 0$ for any $i\in Z$. Then it holds that 
    $$\mc{E}_{G,\vv k}^{\geq 0}=\bigcup_{Z\subseteq[m]} \mc{E}_{G^Z,\vv k^Z}.$$
    
    We claim that for every $Z\subseteq[m]$ such that $\mc{E}_{G^Z,\vv{k}^Z}$ is nonempty,  $Z$ must be a siphon. Indeed, if $Z$ is not a siphon, then $G^Z$ has a reaction where some species $\cf{X}_i$ with $i\in Z$ is formed, but no species $\cf{X}_j$ with $j\in Z$ are reactants, which means that this reaction appears in $G^Z$. Since $G^Z$ by construction does not have any reactions where $\cf{X}_i$ is consumed, this implies that $\mc{E}_{G^Z,\vv k^Z}$ is empty.

    Suppose $Z\subseteq [m]$ is a siphon. Then for every reversible reaction pair $\vv{y}\RR \vv{y}'$ that appears in $G^Z$, it holds that neither $\vv{y}$ nor $\vv{y}'$ involves any species with index in $Z$. Hence, $\widehat{\vv{x}}^Z\in \mc{D}_{G^Z,\vv k^Z}$ for every $\widehat{\vv{x}}\in \mc{D}_{G,\vv{k}}$. This implies that $\mc{D}_{G^Z,\vv{k}^Z}$ is nonempty, and therefore equal to $\mc{E}_{G^Z,\vv{k}^Z}$ by \Cref{thm:DB_properties}\ref{thm-it:DB-syst}. 
    Thus we obtain
    \[
        \mc{E}^{\geq 0}_{G,\vv\kk} 
        = \bigcup_{\substack{ Z \text{ is a} \\ \text{siphon}}} \mc{E}_{G^Z, \vv\kk^Z}
        = \bigcup_{\substack{ Z \text{ is a} \\ \text{siphon}}} \mc{D}_{G^Z, \vv\kk^Z}
        \subseteq \mc{D}^{\geq 0}_{G,\vv\kk}.\qedhere
    \]
\end{proof}

The following extends \Cref{prop:no_boundary_steady_states_imply_global_convergence} to the case when the system is initialized at $\partial\rrp^m$.

\begin{prop}
\label{prop:global_convergence_also_from_boundary_if_no_boundary_steady_states}
Let $(G,\vv{k})$ be a detailed-balanced mass-action system that lacks relevant boundary steady states, and let $\vv{x}_0\in\rrp^m$ be such that $(\vv{x}_0+S)_{>0}$ is nonempty. Then the trajectory $\vv{x}(t;\vv{x}_0)$ converges to the unique positive steady state $\widehat{\vv{x}}$ in $\mc{E}_{G,\vv{k}}\cap(\vv{x}_0+S)_{>0}$
\end{prop}

\begin{proof}
The desired claim follows from \Cref{prop:no_boundary_steady_states_imply_global_convergence} if the trajectory $\vv{x}(t;\vv{x}_0)$ enters the positive orthant at some $t>0$. Suppose for a contradiction that the trajectory is fully contained in $\partial\rrp^m$, and consider the (then nonempty) set
\[Z:=\left\{i\in [m]:\text{$x_i(t;\vv{x}_0)=0$\: for all $t>0$}\right\}.\]
It follows from \cite[Proposition~2]{Angeli2007} that $Z$ is a siphon. Then, using the notation in the proof of \Cref{prop:all_nonnegative_steady_states_are_DB}, $(G^Z,\vv{k}^Z)$ is a detailed-balanced subnetwork, with $\widehat{\vv{x}}^Z$ as a steady state. Since, $\widehat{\vv{x}}^Z$ lies in the same stoichiometric compatibility class as $\vv{x}_0=\vv{x}_0^Z$, this contradicts the assumption that $(G,\vv{k})$ lacks relevant boundary steady states.
\end{proof}

\section{Reaction network representation of log-affine models}\label{sec:main}

{\samepage
The goal of this section is to connect the following two scenarios:
\begin{itemize}
    \item  Algebraic statistics: We saw in \Cref{thm:Birch}  that for a log-affine model $\mc{M}_{\mm A,\vv c}$, the unique MLE for a vector of counts $\vv{u}\in\rrpp^m$ is given by $\mc{M}_{\mm A,\vv c}\cap (\bar{\vv{u}}+\ker(\mm{A}))$. 

    \item Reaction network theory: We saw in  \Cref{thm:DB_properties}
    that for a detailed-balanced system with positive steady state set $\mc{E}_{G,\vv\kk}$ and an initial value $\xx_0 \in \rrpp^m$, the unique positive steady state is given by $\mc{E}_{G,\vv\kk} \cap (\xx_0 + S)$.
\end{itemize} 
More specifically, we will construct a mass-action system for which the set of positive steady states is $\Xac$ and the stoichiometric subspace is $S = \ker(\mm{A})$. Hence, the MLE of the statistical model $\mc M_{\mm A,\vv c}$ for an observed distribution $\bar{\vv{u}}$ will be the unique steady state in $(\bar{\vv{u}} + S)_{>0}$. 
}

A key step in the construction will be to choose a finite set $\Lambda\subseteq\ker_{\zz}(\mm{A})$ that spans $\ker(\mm{A})$ in the vector space sense. This generalizes the construction in \cite[Definition~8]{Gopalkrishnan2016}, which restricts to the case when $\Lambda$ is a $\zz$-basis for $\ker(\mm{A})$ and $\vv{c}=\mathbbm{1}$.
Recall that $\vv\gamma^+$ and $\vv\gamma^-$ are respectively the nonnegative and nonpositive parts of a vector $\vv\gamma\in\zz^m$. 

{\samepage
\begin{cons}\label{const:DB_system}
For a finite set of integer vectors $\Lambda \subseteq \zz^m$ that span $\ker(\mm{A})$, and a vector $\vv{c}\in\rrpp^m$, we define the
mass-action system $G_{\Lambda, \vv c}$ to be the collection
\begin{equation*}
\vcenter{\hbox{\upshape\begin{tikzpicture}
            \begin{scope}[shift={(0,0)}]
            \node (1) at (0,0) {$\gamma^-_1 \cf{X}_1 + \cdots + \gamma^-_m \cf{X}_m$}; 
            \node (2) at (5.5,0) {$\gamma^+_1 \cf{X}_1 + \cdots + \gamma^+_m \cf{X}_m$,}; 
            \draw[fwdrxn, transform canvas={yshift=2pt}] (1)--(2) node[above, midway] {\ratecnst{$\vv c^{\vv \gamma^+}$}};
            \draw[fwdrxn, transform canvas={yshift=-2pt}] (2)--(1) node[below, midway] {\ratecnst{$\vv c^{\vv \gamma^-}$}};
            \node at (8,0) [right] {$\vv\gamma\in \Lambda$\,,};
            \end{scope}
        \end{tikzpicture}}}
\end{equation*}
and denote by $G_{\Lambda}$ the underlying reaction network.
The associated system of ODEs for $G_{\Lambda,\vv{c}}$ is
\begin{equation*}
        \frac{d\xx}{dt} = \sum_{\vv\gamma\in \Lambda} 
            \left( \vv c^{\vv \gamma^+} \xx^{\vv \gamma^-} - \vv c^{\vv \gamma^-} \xx^{\vv \gamma^+} \right)  \vv{\gamma}.
\end{equation*}
\end{cons}
}

We begin by showing that the set of positive steady states of $G_{\Lambda,\vv{c}}$ is $\Xac$, defined in \eqref{eq:define-xac}.

\begin{thm}
\label{thm:realization_of_positive_toric_varieties}
     Let $\mm A \in \zz^{d\times m}$ with $\mathbbm{1}\in\row(A)$, and $\vv c \in \rrpp^m$, and let $\Lambda \subseteq \ker_\zz(\mm{A})$ be a finite spanning set for $\ker(\mm{A})$. Then the mass-action system $G_{\Lambda,\vv{c}}$ in \Cref{const:DB_system} is detailed-balanced, and the set of positive steady states is $\Xac$.
\end{thm}
\begin{proof}  
    The network $G_\Lambda$ is reversible by construction, and the defining condition \eqref{eq:detailed_balanced} for detailed-balanced steady states is   $\vv{c}^{\vv{\gamma}^+}\xx^{\vv{\gamma}^-}=\vv{c}^{\vv{\gamma}^-}\xx^{\vv{\gamma}^+}$ for all $\vv{\gamma}\in\Lambda$. Hence, the  the set of detailed-balanced steady states is $\Xac$ by \eqref{eq:xac-implicit}. Since $\Xac$ is nonempty, it follows that $G_{\Lambda,\vv{c}}$ is detailed-balanced, and $\Xac$ is the set of all positive steady states by \Cref{thm:DB_properties}\ref{thm-it:DB-syst}.
\end{proof}

\begin{rmk}\label{rmk:sufficient_statistics}
By construction, the rows of $\mm A$ form a spanning set for $S^\perp$, and hence, they describe conservation laws (i.e., we can take $\mm W=\mm A$ in the notation from \Cref{sec:RN-intro}), with $\mm{A}\vv{u}$ being the total amounts. 
In statistics, these total amounts are commonly referred to as the \emph{sufficient statistics} of the observed data, where the adjective \emph{sufficient} comes from the fact that to compute the MLE for a given $\vv{u}$, it is enough to know the value of $\mm{A}\vv{u}$.
\end{rmk}

An important goal of \Cref{sec:choosing_spanning_set} will be to study how various dynamical properties of the resulting mass-action system depend on the choice of the spanning set $\Lambda$, with particular focus on vector space bases and Markov bases.
With the following example, we highlight how one can obtain the same MLE but with different reaction networks and dynamics based on different choices of $\Lambda$.

\begin{ex}
\label{ex:dependent}
Consider the log-linear model of \Cref{ex:log_linear_model}, given by 
\[\mm A=\begin{pmatrix}4&2&3&1\\ 0&2&1&3\end{pmatrix}.\]
\begin{enumerate}[label={\normalfont(\alph*)}, leftmargin = 0.7cm]
    \item Consider the spanning set 
$\Lambda_1=\{(1,1,-2,0)^\top, (0,2,-1,-1)^\top\}$. The mass-action system $G_{\Lambda_1}$ given by \Cref{const:DB_system} is \eqref{crn:running_ex-Lambda1} of \Cref{ex:network}, where all the rate constants are $1$ since $\vv c = \mathbbm{1}$, and so $\vv{c}^{\vv \gamma^+} =\vv{c}^{\vv \gamma^-}  = 1$.
\item\label{it:markov} In contrast, the Markov basis 
$\Lambda_\mathrm{mb} = \{(1,1,-2,0)^\top$, $(0,2,-1,-1)^\top$, $(1,-1,-1,1)^\top$, $(2,0,-3,1)^\top\}$  generates the mass-action system $G_{\Lambda_\mathrm{mb}}$ where all rate constants are $1$:
\begin{equation}
\label{crn:running_ex-Markov}
    \vcenter{\hbox{\begin{tikzpicture}
            \begin{scope}[shift={(0,0)}]
            \node (1) at (0,0) [left] {$2\,\cf{X}_2$}; 
            \node (2) at (1.25,0) [right] {$\cf{X}_3 + \cf{X}_4$};
            \draw[fwdrxn, transform canvas={yshift=2pt}] (1)--(2) ;
            \draw[fwdrxn, transform canvas={yshift=-2pt}] (2)--(1) ;
            \end{scope}
            \begin{scope}[shift={(5,0)}]
            \node (1) at (0,0) [left] {$2\,\cf{X}_3$}; 
            \node (2) at (1.25,0) [right] {$\cf{X}_1 + \cf{X}_2$};
            \draw[fwdrxn, transform canvas={yshift=2pt}] (1)--(2) ;
            \draw[fwdrxn, transform canvas={yshift=-2pt}] (2)--(1) ;
            \end{scope}
            \begin{scope}[shift={(0,-0.7)}]
            \node (1) at (0,0) [left] {$\cf{X}_2 + \cf{X}_3$}; 
            \node (2) at (1.25,0) [right] {$\cf{X}_1 + \cf{X}_4$};
            \draw[fwdrxn, transform canvas={yshift=2pt}] (1)--(2) ;
            \draw[fwdrxn, transform canvas={yshift=-2pt}] (2)--(1) ;
            \end{scope}
            \begin{scope}[shift={(5,-0.7)}]
            \node (1) at (0,0) [left] {$3\,\cf{X}_3$}; 
            \node (2) at (1.25,0) [right] {$2\,\cf{X}_1 + \cf{X}_4$.};
            \draw[fwdrxn, transform canvas={yshift=2pt}] (1)--(2) ;
            \draw[fwdrxn, transform canvas={yshift=-2pt}] (2)--(1) ;
            \end{scope}
        \end{tikzpicture}}}
    \end{equation}
The top two reactions make up $G_{\Lambda_1}$, so the right-hand side of the associated ODEs of $G_{\Lambda_\mathrm{mb}}$ can be written as the sum
\eq{ 
\begin{array}{rcll}
    \dot{x}_1 &= & \phantom{-2x_2^2 + 2x_3x_3 +2} x_3^2 - \phantom{2} x_1 x_2 &  {}+ x_2x_3 - x_1 x_4 + 2x_3^3 - 2x_1^2x_4 \\
    \dot x_2 &= & -2x_2^2 + 2x_3x_4 + \phantom{2} x_3^2 - \phantom{2} x_1 x_2 & {}-x_2x_3 + x_1 x_4  \\ 
    \dot x_3 &= & \phantom{-2} x_2^2 - \phantom{2} x_3x_4 - 2x_3^2 + 2x_1x_2 & {}-x_2x_3 + x_1 x_4 - 3x_3^3 + 3x_1^2x_4 \\ 
    \dot x_4 &= & \phantom{-2} x_3^2 - \phantom{2} x_3x_4 & {}+ x_2x_3 - x_1x_4 + \hphantom{2} x_3^3 - \hphantom{2} x_1^2x_4,  \\ 
    & & \raisebox{8pt}{$\underbrace{\hspace{4.75cm}}_{\vv f(\vv x)}$} & \raisebox{8pt}{$\underbrace{\hspace{4.75cm}}_{\vv g(\vv x)}$}
\end{array} 
}
where $\dot{\vv x} = \vv f(\vv x)$ is the ODEs associated to $G_{\Lambda_1}$. 
\item Consider the vector space basis 
$\Lambda_2 = \{(-4,0,6,-2)^\top, \: (3,3,-6,0)^\top\}$. Using this as a spanning set in \Cref{const:DB_system}, gives the following mass-action system where all rate constants are $1$:
\begin{equation*}
    \vcenter{\hbox{\begin{tikzpicture}
            \begin{scope}[shift={(0,0)}]
            \node (1) at (0,0) [left] {$4\, \cf{X}_1 + 2\,\cf{X}_4$} ; 
            \node (2) at (1.75,0) [] {$6\, \cf{X}_3$};
            \node (3) at (3.5,0) [right] {$3\,\cf{X}_1+ 3\, \cf{X}_2$.};
            \draw[fwdrxn, transform canvas={yshift=2pt}] (1)--(2) ;
            \draw[fwdrxn, transform canvas={yshift=-2pt}] (2)--(1) ;
            \draw[fwdrxn, transform canvas={yshift=2pt}] (2)--(3) ;
            \draw[fwdrxn, transform canvas={yshift=-2pt}] (3)--(2) ;
            \end{scope}
        \end{tikzpicture}}}
    \end{equation*}
This is associated to yet again a different system of ODEs: 
    $\dot x_1 =  - 4x_1^4x_4^2  + 7 x_3^6 -3 x_1^3 x_2^3$, 
    $\dot x_2 = 3 x_3^6 -3x_1^3 x_2^3 $, 
    $\dot x_3 = 6x_1^4 x_4^2 - 12 x_3^6 +  6 x_1^3 x_2^3  $, 
    $\dot x_4 = - 2x_1^4x_4^2 + 2x_3^6  $. 
\end{enumerate}
We emphasize that despite the fact that these three networks have different transient dynamics, the sets of positive steady states for all these systems are precisely $\Xa$ in \eqref{eq:log_linear_model-Xac} by \Cref{thm:realization_of_positive_toric_varieties}.
\end{ex}

\begin{rmk}
\Cref{const:DB_system} is not the only approach to constructing a network with $\Xac$ as its positive steady states. For instance, one can decompose each vector $\vv \gamma$ into any difference of nonnegative vectors (not just $\vv\gamma^+$ and $\vv{\gamma}^-$).
Continuing with \Cref{ex:dependent}, we could use $\Lambda_3 = \Lambda_2 \cup \{\vv \mu \} $, where 
$\vv \mu = (4,0,0,2) - (3,3,0,0)$. 
This decomposition gives rise to the network
    \begin{equation*}
    \vcenter{\hbox{\begin{tikzpicture}[scale=1.15]
            \begin{scope}[shift={(0,0)}]
            \node (1) at (0,0) [left] {$4\,\cf{X}_1 + 2\, \cf{X}_4$}; 
            \node (2) at (1.5,0) [right] {$3\, \cf{X}_1 + 3\, \cf{X}_2$\rlap{\,.}}; 
            \node (3) at (0.75,-1.5) {\,$6\,\cf{X}_3$} ;
            \node (3r) at (0.85,-1.35) {};
            \node (3l) at (0.65,-1.35) {}; 
            \draw[fwdrxn, transform canvas={yshift=2pt}] (1)--(2) ;
            \draw[fwdrxn, transform canvas={yshift=-2pt}] (2)--(1) ;
            \draw[fwdrxn, transform canvas={xshift=1.5pt, yshift=-1.5pt}] (2.south west)--(3r);
            \draw[fwdrxn, transform canvas={xshift=-1.5pt, yshift=1.5pt}] (3r)--(2.south west) ;
            \draw[fwdrxn, transform canvas={xshift=-1.55pt, yshift=-1.45pt}] (3l)--(1.south east) ;
            \draw[fwdrxn, transform canvas={xshift=1.55pt, yshift=1.45pt}] (1.south east)--(3l) ;
            \end{scope}
        \end{tikzpicture}}}
    \end{equation*}
\end{rmk}

The following is one of our  main results, which states that the MLEs over the model $\mc{M}_{\mm A,\vv c}$ are realized as steady states of the mass-action systems from \Cref{const:DB_system}.

\begin{thm}
\label{thm:main}
    Let $\mc{M}_{\mm A,\vv c}$ be the log-affine model defined by $\mm{A}\in\zz^{d\times m}$ with $\mathbbm{1}\in\row(\mm{A})$, and $\vv c \in \rrpp^m$. Let $\vv{u}\in\rrp^m\setminus\{\vv{0}\}$, and $\bar{\vv u} = \vv u /u_+$, and suppose  $(\bar{\vv{u}}+\ker(\mm{A}))_{>0}$ is nonempty. Let $\Lambda \subseteq \ker_\zz(\mm A)$ be a finite spanning set for $\ker(\mm{A})$, and let $G_{\Lambda, \vv c}$ be the mass-action system in \Cref{const:DB_system}. 
        Then the MLE $\widehat{\vv p}$ of $\vv{u}$ in $\mc{M}_{\mm A,\vv c}$ is the unique steady state of $G_{\Lambda,\vv{c}}$ in the positive stoichiometric compatibility class $(\bar{\vv u}+S)_{>0}$. Furthermore, $\widehat{\vv p}$ is asymptotically stable.
\end{thm}
\begin{proof}
By \Cref{thm:realization_of_positive_toric_varieties}, the system $G_{\Lambda,\vv{c}}$ is detailed-balanced, and the positive steady states are given by $\Xac$. The positive stoichiometric compatibility class of $\bar{\vv{u}}$ is $(\bar{\vv{u}}+\ker(\mm{A}))_{>0}$ by \Cref{rmk:sufficient_statistics}, so the unique steady state in this class is precisely the MLE $\widehat{\vv p}$ by \Cref{thm:Birch}. Asymptotic stability follows by \Cref{thm:DB_properties}\ref{thm-it:DB-stability}.
\end{proof}

\begin{ex}[The Hardy--Weinberg model revisited]
\label{ex:trajectory_binomial_2_theta}
Consider the model from \Cref{ex:hardy_weinberg}. Using the Markov basis
$\Lambda=\{(1,-2,1)\}$, we obtain precisely the mass-action system that appeared in \Cref{ex-DYN-stoichiometric_compatibility_class-a}. For a given initial state $\vv{x}_0=(u_1,u_2,u_3)$, the initial value problem 
has the closed form solution
\begin{align*}
    x_1(t) &= \frac{  (2u_1+u_2)^2 - (u_2^2-4u_1u_3)e^{-4t} }{4(u_1+u_2+u_3)^2}\\ 
    x_2(t) &= \frac{  (u_2+2u_3)(2u_1+u_2) + (u_2^2-4u_1u_3)e^{-4t}}{2(u_1+u_2+u_3)^2} \\ 
     x_3(t) &= \frac{  (u_2+2u_3)^2 - (u_2^2-4u_1u_3)e^{-4t} }{4(u_1+u_2+u_3)^2}.
\end{align*}
Letting $t \to \infty$, we recover the rational expression for the MLE from \eqref{eq:hardy_weinberg_mle}:
\[ \lim_{t\to\infty} \xx(t) =   \left( \frac{(2u_1+u_2)^2}{4(u_1+u_2+u_3)^2} , \frac{(u_2+2u_3)(2u_1+u_2)}{2(u_1+u_2+u_3)^2}  , \frac{(u_2+2u_3)^2}{4(u_1+u_2+u_3)^2} \right) = \widehat{\vv p}.\] 
Note also that if $\vv u \in \mathcal{M}_{\mm A, \vv c}$, then $\vv x(t)=\bar{\vv u}$ for all $t \geq 0$, i.e., it remains constant as the MLE.
\end{ex}

We close the section with two remarks about possible variations of \Cref{const:DB_system}. 

\begin{rmk}\label{rmk:quadratic}
Networks that appear in biological systems are typically \emph{bimolecular}, in the sense that the coefficients of each reactant complex sum up to at most 2. The maximum-likelihood-estimating networks of \Cref{const:DB_system} are not necessarily bimolecular, which could pose a challenge in their biochemical implementation. However, we note that they can always be modified into bimolecular networks after introducing additional auxiliary species and reactions, in such a way that the steady state concentrations in the original species still is the desired MLE~\cite{Kerner1981-recasting,HERNANDEZBERMEJO199531,Wilhelm2000, Plesa2023}. 
\end{rmk}

\begin{rmk}\label{rmk:augmented_system}
This section has been focused on proving that for a given log-affine model $\mc M_{\mm A,\vv c}$, one can compute the MLE $\widehat{\vv{p}}$ using the mass-action system $G_{\Lambda, \vv c}$ in \Cref{const:DB_system}.
In some applications, it could be useful to also obtain the parameters $\vv{t}=(t_1,\ldots,t_d)$ that give rise to $\widehat{\vv{p}}$ under the parametrization $\varphi_{\mm A,\vv c}$ from \eqref{eq:parametrization_of_log_affine_set}. In order to do this, \cite{Gopalkrishnan2016} proposed an extension to $G_{\Lambda, \vv c}$; the following is based on the same idea. 

Choose a set $J\subseteq[m]$ of linearly independent columns of $\mm A$ that span the column space. For each $i \in [d]$, add a new species $\cf{Z}_i$; moreover for each $j \in J$, add the reactions
\begin{equation*}
    \begin{tikzpicture}
        \node (1a) at (0,0) [left] {$\ds \sum_{i=1}^d a_{ij} \cf{Z}_i$};
        \node (2a) at (1,0) [right] {$\cf{0}$,\hspace{1.2em} and };
        \draw[fwdrxn, transform canvas={yshift=0pt}] (1a) -- (2a) node[above, midway] {\ratecnst{{$c_j$}}};

        \node (1b) at (3.75,0) [left] {$\cf{X}_j$}; 
        \node (2b) at (4.75,0) [right] {$\ds \cf{X}_j + a_{ij} \cf{Z}_i$};
        \draw[fwdrxn, transform canvas={yshift=0pt}] (1b) -- (2b) node[above, midway] {\ratecnst{{$1$}}};
        
        \node at (11.5,0) [left] {for each $i$ with $a_{ij} > 0$.};
    \end{tikzpicture}
\end{equation*}
Notice that the added reactions do not affect any of the equations $\dot{x}_j$; hence the subsystem for $\dot{x}_j$ still converges to the MLE $\widehat{\vv x} = \widehat{\vv p}$. The added variables $z_i$ evolve according to 
$\dot{\vv z}(t) = \sum_{j \in J}  \left( x_j(t) - c_j \vv z(t)^{\vv a_j} \right)\vv{a}_j$,
which at steady state satisfies $\widehat{x}_j = c_j \widehat{\vv z}^{\vv a_j}$ for all $j \in J$. 
Hence, at steady state, we recover the parametrization $\widehat{\vv x} = \vv c \circ \vv t^{\mm A}$ for the subsystem on $\xx$. An analysis similar to that of \cite[Theorem 6]{Gopalkrishnan2016} shows that the augmented system is globally stable if and only if $G_{\Lambda,\vv{c}}$ is globally stable with respect to the appropriate stoichiometric compatibility class. 
\end{rmk}

\section{Impact of the spanning set on the dynamics}
\label{sec:choosing_spanning_set}

A key step in \Cref{const:DB_system} is choosing a finite spanning set $\Lambda\subseteq\ker_{\zz}(\mm{A})$ for $\ker(\mm{A})$. In this section, we discuss how this choice influences the dynamical properties of the mass-action system $G_{\Lambda, \vv c}$. We will focus on three properties: global stability (\Cref{sec:boundary}), deficiency (\Cref{sec:DZ}), and the rate of convergence (\Cref{sec:slow-mfld}). The two main options we will consider are the following:
\begin{itemize}
    \item If $\Lambda$ is chosen to be a \emph{Markov basis}, we show in \Cref{sec:boundary} that asymptotic stability can be strengthened to global stability, and in \Cref{sec:slow-mfld}, we give examples where it ensures fast and uniform convergence.
    \item If $\Lambda$ is a \emph{vector space basis} for $\ker(\mm{A})$, the cardinality is minimal, which means the fewest possible reactions to be engineered in a synthetic biology context. In \Cref{sec:DZ}, we show that this implies that the deficiency is zero.
\end{itemize}

\subsection{Global stability and boundary steady states}
\label{sec:boundary}

An important question about our maximum-likelihood-estimating networks from the point of view of computational reliability is whether we have \emph{global stability}, in the sense that the trajectory of the initial value problem in \Cref{thm:main} converges to the MLE $\widehat{\vv{p}}$ for any vector of counts $\vv{u}\in\rrp^m$. 

As pointed out in \Cref{sec:RN-intro}, the Global Attractor Conjecture states that we have global convergence for strictly positive vectors $\vv{u}\in\rrpp^m$, but is only proven for special classes of networks, and does not directly address the case when $\vv{u}\in\partial\rrp^m$. 

In this section, we will focus on the property of \emph{lacking relevant boundary steady states} (discussed at the end of \Cref{sec:RN-intro}) as a sufficient condition for global convergence. For $G_{\Lambda,\vv{c}}$, this is equivalent to the system
\[
\sum_{\vv\gamma\in \Lambda} \big( \vv c^{\vv \gamma^+} \xx^{\vv \gamma^-} - \vv c^{\vv \gamma^-} \xx^{\vv \gamma^+}) \cdot \vv\gamma
=\vv 0,
\qquad 
\mm{A}\vv x=\mm{A}\vv u\]
having no solution in $\partial\rrp^m$ for any $\vv u\in\rrpp^m$.

\begin{thm}
\label{thm:global_convergence_to_the_MLE_if_no_relevant_boundary_steady_states}
Let $\mm{A}\in\zz^{d\times m}$ with $\mathbbm{1}\in\row(\mm{A})$, and $\vv{c}\in\rrpp^m$, and suppose that $\Lambda\subseteq\ker_{\zz}(\mm{A})$ is a finite spanning set of $\ker(\mm{A})$ such that $G_{\Lambda,\vv{c}}$ lacks relevant boundary steady states.  Then, for any vector of observed data $\vv{u}\in\rrp^m$ such that $(\bar{\vv{u}}+\ker(\mm{A}))_{>0}$ is nonempty, the trajectory $\vv{x}(t;\bar{\vv{u}})$ converges to the MLE of $\vv{u}$ in the model $\mc{M}_{\mm{A},\vv{c}}$.
\end{thm}

\begin{proof}
This is a direct consequence of \Cref{prop:global_convergence_also_from_boundary_if_no_boundary_steady_states} and \Cref{thm:main}.
\end{proof}

It is natural to ask whether the spanning set $\Lambda$ can be chosen in a way that rules out relevant boundary states, and thereby ensures global convergence. We prove that choosing $\Lambda$ to be a Markov basis achieves this.

\begin{thm}\label{thm:markov_basis_implies_no_relevant_boundary_steady_states}
If $\Lambda\subseteq\ker_{\zz}(\mm{A})$ is a Markov basis for $\mm{A}$, then the mass-action system $G_{\Lambda, \vv c}$ in \Cref{const:DB_system} lacks relevant boundary steady states. In particular, $G_{\Lambda, \vv c}$ satisfies the global convergence property of \Cref{thm:global_convergence_to_the_MLE_if_no_relevant_boundary_steady_states}.
\end{thm}

\begin{proof}
By \Cref{prop:all_nonnegative_steady_states_are_DB}, the set of nonnegative steady states of $G_{\Lambda,\vv{c}}$ is given by
    \[\mc{E}_{G_{\Lambda,\vv{c}}}^{\geq 0}=\Big\{\vv{x}\in\rrp^m:\vv{c}^{\vv{\gamma}^-}\vv{x}^{\vv{\gamma}^+}=\vv{c}^{\vv{\gamma}^+}\vv{x}^{\vv\gamma^-}\text{ for all }\vv\gamma\in\Lambda\Big\}=\closedXac,\]
    where the second equality follows from \Cref{lem:markov_basis_cuts_out_closedXac}. The lack of relevant boundary steady states now follows from the fact that, for any $\vv{u}\in\rrpp^m$, the intersection $\closedXac\cap(\vv{u}+\ker(\mm{A}))$ consists of a strictly positive point (and therefore has empty intersection with $\partial\rrp^m$) according to \Cref{thm:Birch}.
\end{proof}

\begin{rmk}
A similar statement appears in \cite[Lemma~2]{Gopalkrishnan2016} for $\vv c = \mathbbm{1}$, with a different proof relying on siphon theory. The statement of \cite{Gopalkrishnan2016} uses the weaker assumption that $\Lambda$ is a $\zz$-basis for the abelian group $\ker_{\zz}(\mm{A})$, rather than a Markov basis, but as shown by \Cref{ex:boundary_steady_states}\ref{it:counterexample} below, this is \emph{not} sufficient to rule out relevant boundary steady states. Indeed, the key step in the proof given in \cite{Gopalkrishnan2016} is to show that the so-called \emph{associated ideal} $J \coloneqq \langle \vv{x}^\vv{y}-\vv{x}^{\vv y'} :  \vv y\to \vv y'\in E \rangle$ of the network is prime, which proves structural persistence by \cite[Theorems~4.1 and~5.2]{gopalkrishnan2011catalysis}, and thereby rules out relevant boundary steady states. For $G_{\Lambda}$, we have $J=\langle \vv{x}^{\vv \gamma^+}-\vv{x}^{\vv \gamma^-} :  \vv\gamma\in\Lambda \rangle$.  When $\Lambda$ is a Markov basis, primeness of $J$ follows from the fact that toric varieties are irreducible, but when $\Lambda$ is not a Markov basis, $J$ is always non-prime.
\end{rmk}

If $\Lambda\subseteq\ker(\mm{A})$ is not a Markov basis, there might be additional solutions to the binomial equations $\vv{c}^{\vv{\gamma}^-}\vv{x}^{\vv{\gamma}^+}=\vv{c}^{\vv{\gamma}^+}\vv{x}^{\vv\gamma^-}$ for $\vv\gamma\in\Lambda$ 
compared to the binomial system induced by a Markov basis for $\mm{A}$. However, it is only in some cases that these additional solutions constitute relevant boundary steady states. We illustrate this, as well as the concept of siphons and structural persistence, in the next example. Note that for a network $G_{\Lambda}$, a siphon corresponds to a set $Z\subseteq [m]$ such that for every $\vv{\gamma}\in\Lambda$ and $i\in Z$ with $\gamma_i>0$, there is a $j\in Z$ such that $\gamma_j<0$. Structural persistence means that for every minimal siphon $Z$, there is some $\vv a\in\row(\mm{A})\cap\rrp^m$ with $\supp(\vv a)=Z$.

\begin{ex}\label{ex:boundary_steady_states}
Consider again the log-linear model $\mathcal{M}_{\mm A}$ from \Cref{ex:log_linear_model,ex:dependent}, where 
\[\mm{A}=\begin{pmatrix}4&2&3&1\\ 0&2&1&3\end{pmatrix}.\]
\begin{enumerate}[label=(\alph*), leftmargin = 0.7cm]
\item\label{it:counterexample} Recall the non-Markov basis $\Lambda_1 =\{(1,1,-2,0)^\top, (0,2,-1,-1)^\top\}$, which is a $\zz$-basis for $\ker_{\zz}(\mm{A})$. From this, we obtain the mass-action system in \eqref{crn:running_ex-Lambda1}, where all the rate constants are $1$. This system has relevant boundary steady states, e.g., $\widehat{\xx} = (0.9,0,0,0.1)^\top$; see red box in \Cref{fig:ex:boundary_steady_states-vsp}.
In this case, $\{2,3\}$ is a minimal siphon, but there are no vectors $\vv{w}\in\row(\mm{A})$ with $\supp(\vv{w})=\{2,3\}$, so the siphon criterion is inconclusive.

\item For the Markov basis 
$\Lambda_\mathrm{mb} =\{(1,1,-2,0)^\top$, $(0,2,-1,-1)^\top$, $(2,0,-3,1)^\top$, $(1,-1,-1,1)^\top\}$, we obtained the mass-action system \eqref{crn:running_ex-Markov} with four pairs of reversible reactions. By \Cref{thm:markov_basis_implies_no_relevant_boundary_steady_states}, $G_{\Lambda}$ lacks relevant boundary steady states, and we are guaranteed to have global convergence. This is illustrated in \Cref{fig:ex:boundary_steady_states-Markov}. 

\item\label{it:vector_space_basis_with_persistence} Finally, for the non-Markov basis $\Lambda_4=\{ (1,-3,0,2)^\top, (-1,-1,2,0)^\top\}$, which is also a $\zz$-basis for $\ker_{\zz}(\mm{A})$, we obtain a structurally persistent network; the minimal siphons are $\{1,2,3\}$ and $\{2,3,4\}$, which are the supports of the vectors
$(12,4,8,0)^\top,(0,2,1,3)^\top\in\row(\mm{A})$. Hence, it follows from \Cref{prop:no_boundary_steady_states_imply_global_convergence,prop:siphon_critierion} that we have global convergence.
\end{enumerate}
\end{ex}

\begin{figure}[t]
\centering
    \begin{subfigure}[t]{0.48\textwidth}
    \centering 
        \includegraphics[width=2.5in]{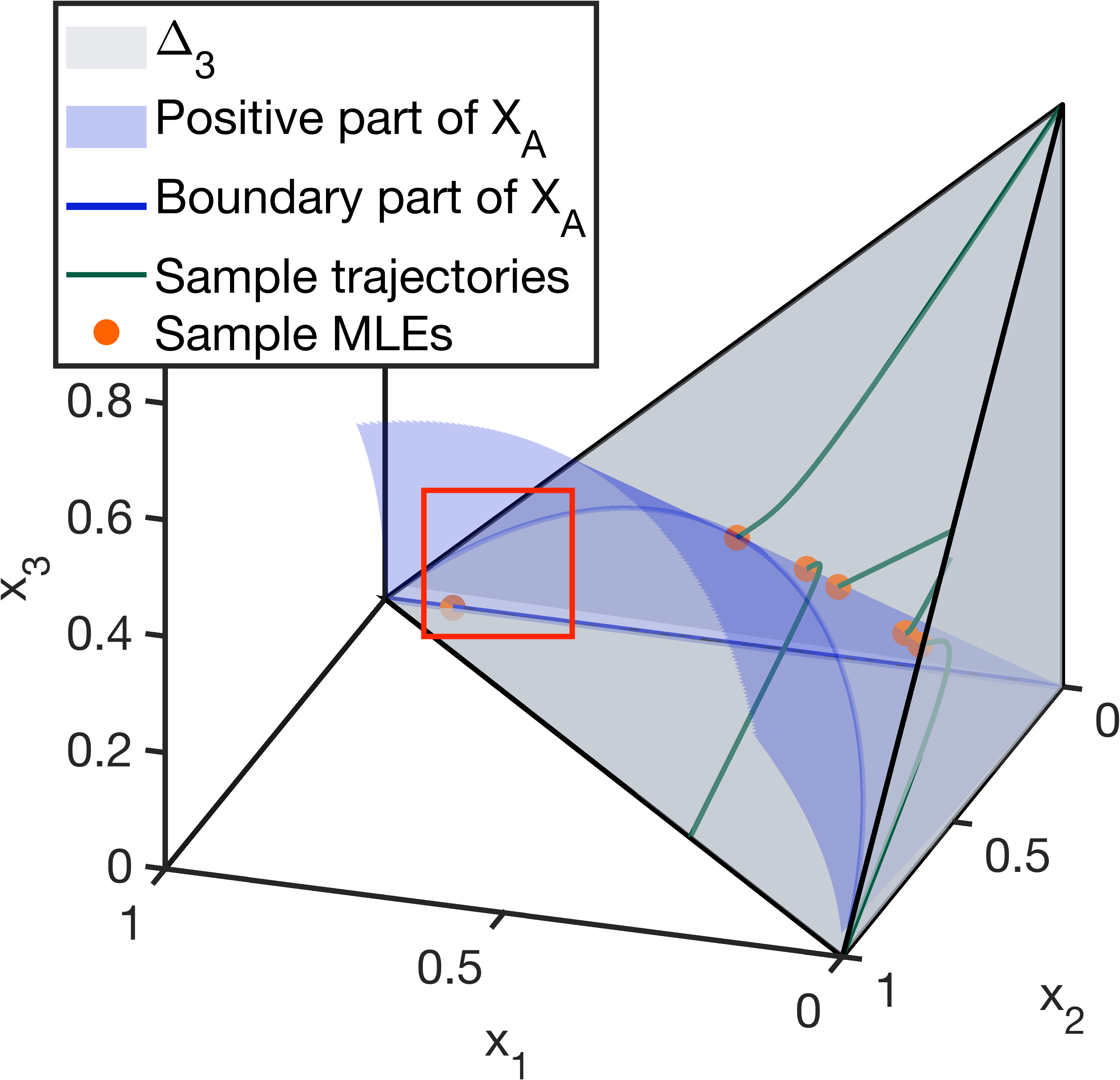}
        \caption{$\Lambda_1$}
        \label{fig:ex:boundary_steady_states-vsp}
    \end{subfigure}
    \begin{subfigure}[t]{0.48\textwidth}
    \centering 
        \includegraphics[width=2.5in]{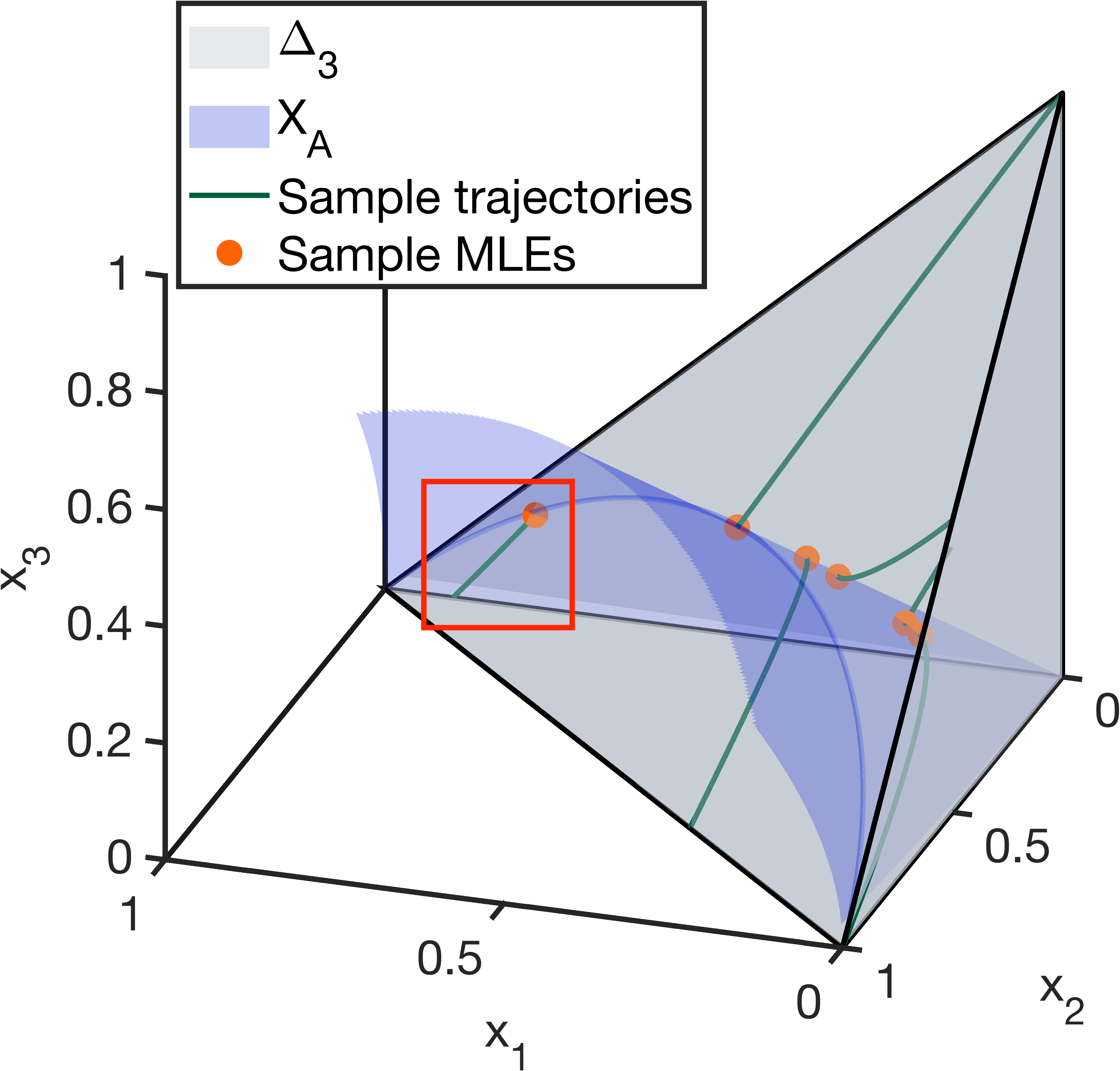}
        \caption{$\Lambda_\mrm{mb}$}
        \label{fig:ex:boundary_steady_states-Markov}
    \end{subfigure}
    \caption{Illustration of \Cref{ex:boundary_steady_states}, showing trajectories and steady states for several initial conditions under the canonical projection to  $(x_1,x_2,x_3)$. Red~boxes: The observed distribution $\bar{\vv u} = (0.9,0,0,0.1)^\top$ is a boundary steady state for system given by $\Lambda_1$, while the system given by $\Lambda_\mrm{mb}$ evolves from $\bar{\vv u}$ to the MLE.}
    \label{fig:ex:boundary_steady_states}
\end{figure}

\subsection{Deficiency}
\label{sec:DZ}

A challenge in the chemical implementation of $G_{\Lambda,\vv{c}}$ is fine-tuning the rate constants to the values prescribed by \Cref{const:DB_system}. Because of this, an important question is to what extent the dynamical properties of $G_{\Lambda,\vv{c}}$ are preserved under perturbations to the rate constants. 

In this section, we will focus on the particularly well-behaved scenario of deficiency zero. The \demph{deficiency} of a (weakly) reversible reaction network $G=(V,E)$ is $\delta := |V| - \ell - \dim(S)$, where $\ell$ is the number of connected components and $S$ is the stoichiometric subspace. By \cite{CDSS09}, $\delta$ is the codimension of the subset of the rate constants for which $G$ is complex-balanced (and hence satisfies analogous properties as those given in \Cref{thm:DB_properties,prop:global_convergence_also_from_boundary_if_no_boundary_steady_states}). In particular, if $\delta=0$, the Deficiency Zero Theorem \cite{Horn1972} states that the network is complex-balanced for \emph{all} values of the rate constants. For a more in-depth introduction to deficiency theory, we refer to  \cite{Fei19}.

\begin{thm}
\label{thm:deficiency-zero}
Let $\mm{A}\in\zz^{d\times m}$ and let $\Lambda\subseteq\ker_{\zz}(\mm{A})$ be a finite spanning set of $\ker(\mm{A})$. 
\begin{enumerate}[label=(\alph*)]
\item The deficiency of $G_{\Lambda}$ is $\delta=|\Lambda|-r-\dim(\ker(\mm{A}))$, where $|\Lambda|$ is the cardinality of $\Lambda$ and $r$ is the cycle rank (the number of independent cycles) in the underlying undirected graph of $G_{\Lambda}$.
\item If $\Lambda$ is a vector space basis for $\ker(\mm{A})$, then the undirected graph of $G_{\Lambda}$ is acyclic and $\delta=0$.
\end{enumerate}
\end{thm}
\begin{proof}
    Part (a) follows from the general graph-theoretic fact that the number of edges minus the cycle rank equals the number of vertices minus the number of connected components. Part (b) follows from part (a), combined with the fact that $\delta$ cannot be negative.
\end{proof}

Note that $G_{\Lambda}$ has deficiency zero precisely when the cycle rank equals the number of vectors in $|\Lambda|$ that are redundant for spanning $\ker(\mm{A})$.
When this property of having deficiency zero coincides with the property of lacking relevant boundary steady states discussed in \Cref{sec:boundary} (e.g., when $\Lambda$ is both a vector space basis and Markov basis), then the network displays very robust dynamics. 

\begin{cor}
\label{cor:dz_and_global_convergence}
Suppose that $\Lambda\subseteq\ker_{\zz}(\mm{A})$ is such that $\delta=0$ and $G_{\Lambda}$ lacks relevant boundary steady states for all rate constants. Then $G_{\Lambda}$ has a unique and globally stable positive steady state in each stoichiometric compatibility class for every choice of rate constants. 
\end{cor}
\begin{proof}
This follows by the complex-balanced analog of \Cref{prop:global_convergence_also_from_boundary_if_no_boundary_steady_states}, combined with the Deficiency Zero Theorem \cite[Theorem~4A]{Horn1972}.
\end{proof}

\begin{ex}
Consider again the choices of spanning sets discussed in \Cref{ex:boundary_steady_states}. Using $\Lambda_1$ gives a network with $\delta=0$, $\Lambda_{\operatorname{mb}}$ gives a network with $\delta=1$, and $\Lambda_4$ gives a network with $\delta=0$ that satisfies the property in \Cref{cor:dz_and_global_convergence}.
\end{ex}

\subsection{Rate of convergence}
\label{sec:slow-mfld}
In the neighborhood of a hyperbolic equilibrium, the rate of convergence is dictated by the eigenvalue $\lambda_{\max}$ with the least negative real part for the Jacobian matrix~\cite{ChiconeODE}, assuming all its eigenvalues have negative real parts. Convergence along the eigendirection for $\lambda_{\max}$ is the slowest; thus the more negative $\mathrm{Re}(\lambda_{\max})$ is, the faster the system converges to equilibrium, with a rate of convergence $-\mathrm{Re}(\lambda_{\max})$. For detailed-balanced systems, we only consider its nonzero eigenvalues, since the detailed-balanced steady state is hyperbolic with respect to its stoichiometric compatibility class~\cite{Johnston2008}. 

In this section, we explore how the rate of convergence depends on the choice of the spanning set, using as two examples: the binary 4-cycle (\Cref{ex:4cycle}), and the independence model with design matrix $\mm A_\mathrm{ind}(10,10)$ (\Cref{ex:independence}). 
(Similar observations were made for $\mm A_\mathrm{ind}(r_1,r_2)$ with $3 \leq r_1 \leq 10$ and $2 \leq r_2 \leq 10$, and for scaling factors other than $\mathbbm{1}$.) Each of these models has the property that there is a unique Markov basis $\Lambda_\mathrm{mb}\subseteq\ker_{\zz}(\mm{A})$ of minimal cardinality; see the database \texttt{Markov-Bases.de} \cite{4ti2} and \cite[Proposition~1.2.2]{drton2008lectures}.

In our experiments, we first constructed a strictly increasing sequence of subsets
\begin{equation*}
    \Lambda_1\subsetneq\Lambda_2 \subsetneq\cdots \subsetneq 
    \Lambda_9\subsetneq \Lambda_{10}
\end{equation*}
where $\Lambda_1$ is a vector space basis for $\ker(\mm{A})$, and $\Lambda_{10}$ has the same cardinality as  $\Lambda_\mathrm{mb}$. Then for each $\Lambda_i$ and for a random choice of initial conditions $\vv{u} \in \Delta_{m-1}$, we approximated the positive steady state, i.e., the MLE, and the nonzero eigenvalues of the Jacobian matrix. Thus we estimated the rate of convergence $-\mathrm{Re}(\lambda_{\max})$ of the system $G_{\Lambda_i, \mathbbm{1}}$. We conducted two numerical experiments:  (i) where $\Lambda_{10}$ is the unique Markov basis of minimal cardinality, and (ii) where $\Lambda_{10}$ is a random spanning set%
	\footnote{
		Obtained by extending $\Lambda_1$ with linear combinations of vectors in $\Lambda_1$ whose coefficients were uniformly drawn from $\{-1,0,1\}$.
	}%
. 
Detailed numerical results and codes for the numerical experiments are available at 
\url{https://github.com/oskarhenriksson/estimating-birch-points-with-networks}.%
    \footnote{The code is based on Python 3.10.12, with packages NumPy 1.26.4~\cite{numpy}, Scipy 1.11.4~\cite{scipy}, and the reaction networks theory package Tellurium 2.2.10~\cite{ChoiMedleyKonigStockingEtAl2018}.}

\begin{figure}[h!bt]
\centering 
\begin{subfigure}[t]{0.48\textwidth}
\centering 
\includegraphics[width=0.95\textwidth]{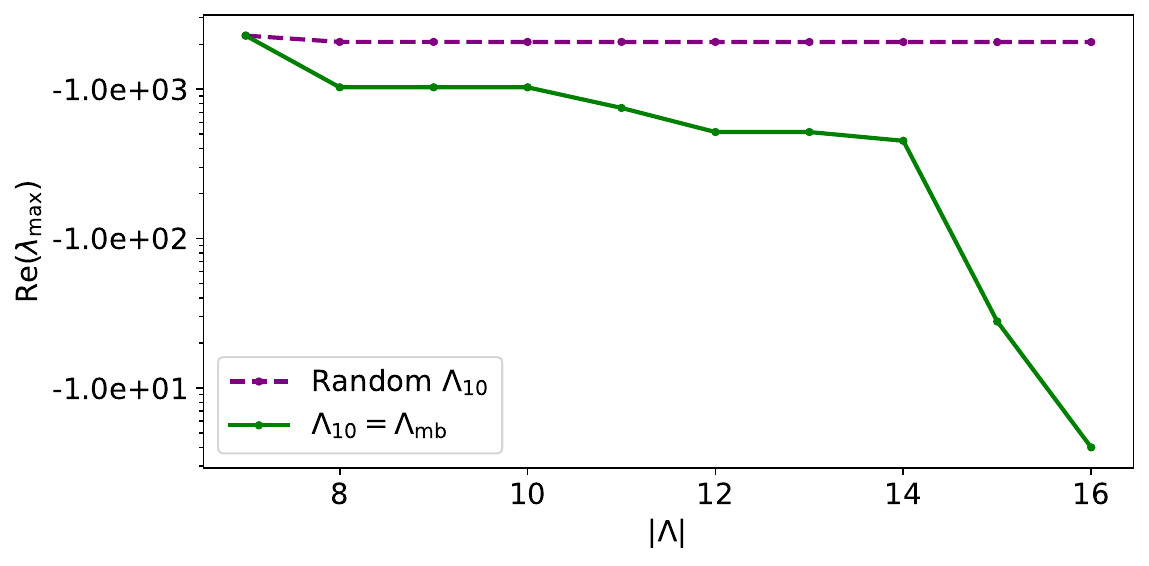}
    \caption{}
    \label{fig:RateCnvg-4cycle}
\end{subfigure} 
\begin{subfigure}[t]{0.48\textwidth}
\centering
\includegraphics[width=0.95\textwidth]{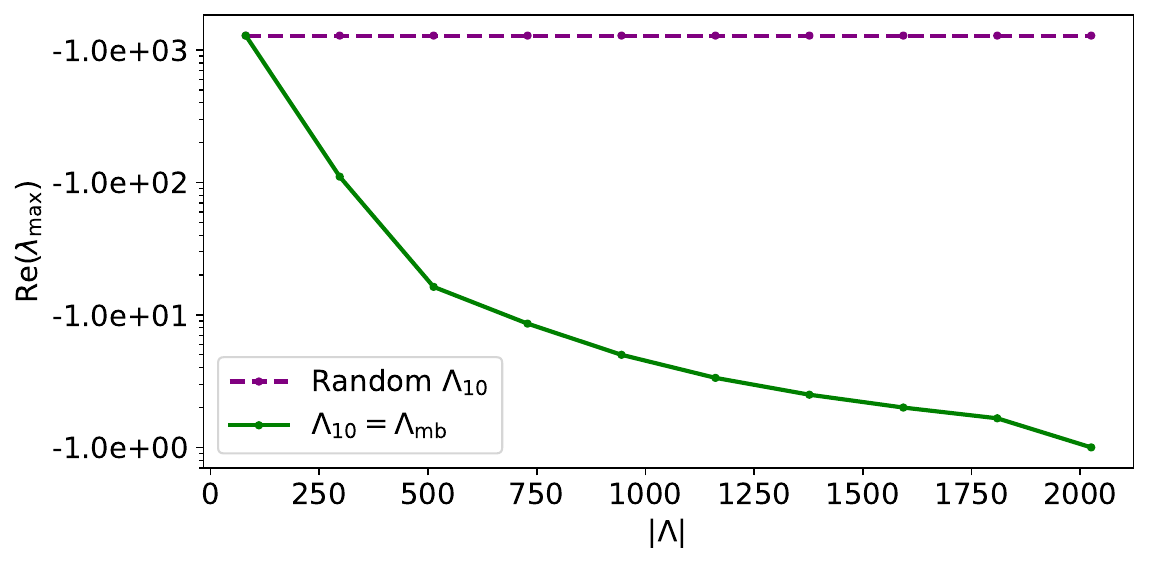}
    \caption{}
    \label{fig:RateCnvg-indep}
\end{subfigure} 
\caption{The least negative nonzero eigenvalues for systems constructed using increasing sequences of spanning sets $\Lambda_i$, which increase towards either (solid) the unique Markov basis $\Lambda_{10} = \Lambda_{\mathrm{mb}}$, or (dashed) a random spanning set. The models considered were (a)  the binary 4-cycle model (\Cref{ex:4cycle}), and (b) the independence model (\Cref{ex:independence}) with $r_1 = 10$ and $r_2 = 10$. }
\label{fig:rate_cnvg}
\end{figure}

\Cref{fig:rate_cnvg} is a plot of the real parts of the least negative nonzero eigenvalues $\mathrm{Re}(\lambda_{\max})$ for each $\Lambda_i$, where the results for the binary 4-cycles are shown in (a), and those for the independence model are shown in (b). Notice that for each model, for a fixed number of reactions, i.e., for fixed $|\Lambda_i| $, the rate of convergence is faster when $\Lambda_i$ is increasing towards a Markov basis, compared to when they came from a random spanning set. The numerical experiments suggest that one cannot speed up the rate of convergence by simply adding more vectors to the spanning set, unless the vectors are chosen wisely, e.g., building up towards a Markov basis. Understanding this interplay between algebraic geometry and dynamics is an interesting direction for future research.

When the spanning set $\Lambda_i$ increases towards the Markov basis, we observed in these two examples that not only is the least negative nonzero eigenvalue monotonically decreasing (i.e., increasing rate of convergence), but that \emph{all} eigenvalues tend to be more negative. \Cref{fig:all_evals-markov} shows all the nonzero eigenvalues for (a) the binary 4-cycle model and (b) independence model. The phenomenon is more pronounced in the case of the independence model (\Cref{fig:all_evals-indep}), where we observed that all the eigenvalues are $-1$ for the minimal Markov basis, i.e., the system converges to the MLE globally with rate of convergence $1$.

\begin{figure}[h!tb]
\centering 
\begin{subfigure}[t]{0.48\textwidth}
\centering 
\includegraphics[width=0.95\textwidth]{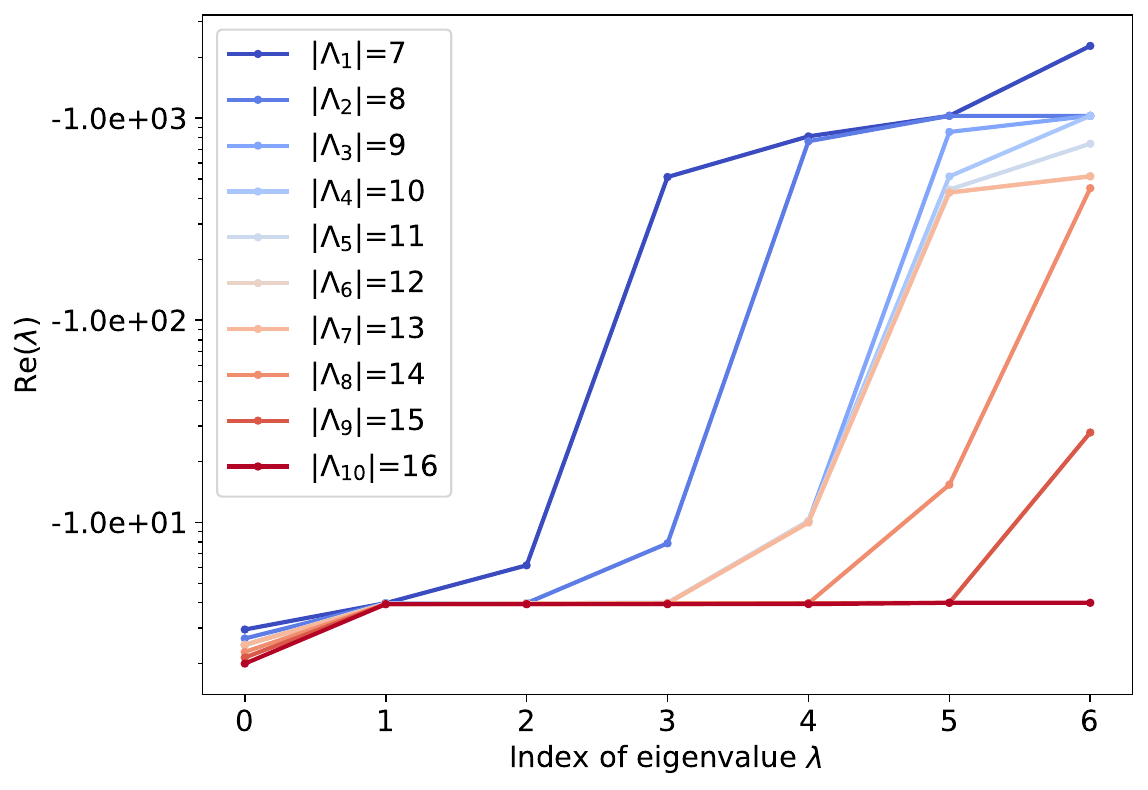}
    \caption{}
    \label{fig:all_evals-4cycle}
\end{subfigure} 
\begin{subfigure}[t]{0.48\textwidth}
\centering
\includegraphics[width=0.95\textwidth]{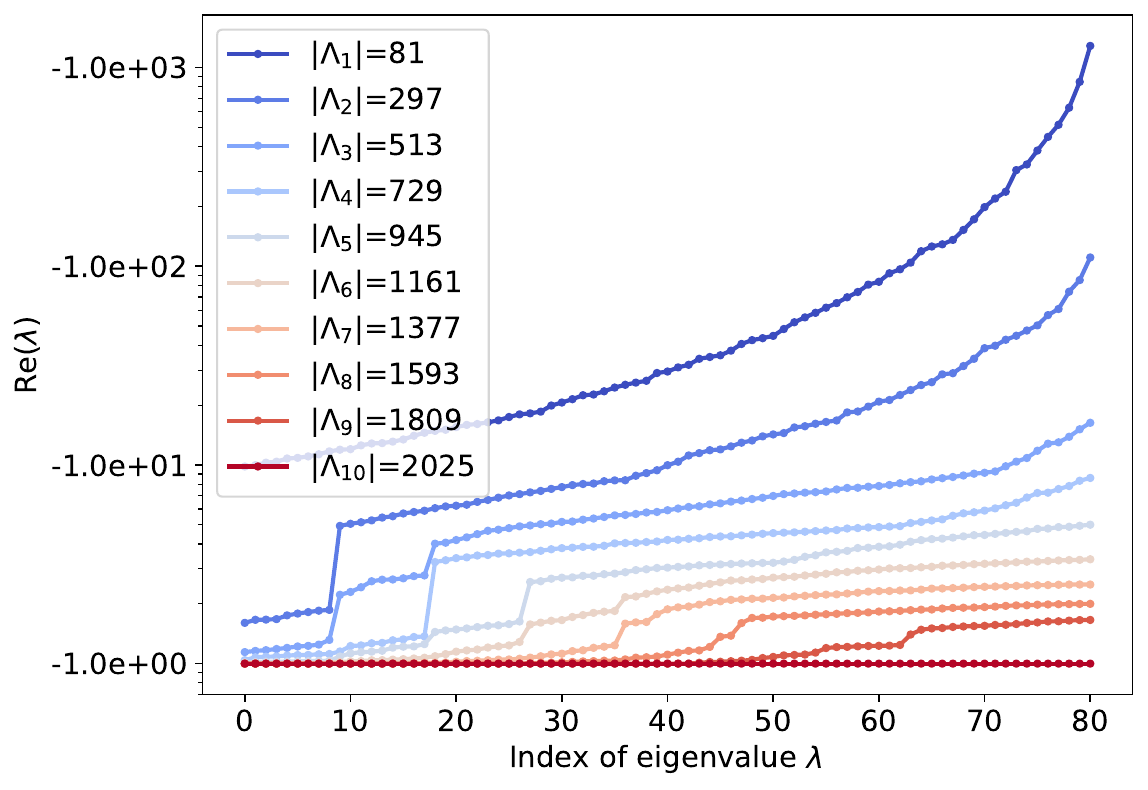}
    \caption{}
    \label{fig:all_evals-indep}
\end{subfigure} 
\caption{All nonzero eigenvalues of the systems constructed using $\Lambda_i$, which increases towards the Markov basis $\Lambda_{10}$ for (a)  the binary 4-cycle model (\Cref{ex:4cycle}), and (b) the independence model (\Cref{ex:independence}) with $r_1 = 10$ and $r_2 = 10$. As $\Lambda_i$ increases towards the Markov basis, all nonzero eigenvalues tend to become more negative. 
}
\label{fig:all_evals-markov}
\end{figure} 

We conclude this section by proving that the numerical observation concerning the independence model holds more generally:  the detailed-balanced system constructed using the unique minimal Markov basis always has global rate of convergence $1$. In fact, we give explicit analytical solutions for the independence model, which shows that the rate of convergence is \emph{globally} exponential with rate $1$, not just near the MLE. This is a highly desirable feature from a molecular computation point of view, since it means that the system will converge to the MLE with a rate that is independent of the input (see, e.g., the discussion in \cite{AndersonJoshi2024}). See also \Cref{ex:trajectory_binomial_2_theta} for another example of when a Markov basis gives rise to a global rate of convergence.  An interesting problem for future work is to investigate whether we can always find closed form formulas for models where the MLE depends rationally on the observed distribution.

\begin{prop}\label{prop:anlaytic_solution_for_independence_models_with_markov_basis}
Consider the initial value problem constructed by \Cref{const:DB_system} using the unique minimal Markov basis \eqref{eq:markov_indep} for the independence model with $r_1$ and $r_2$ states:
\begin{equation}\label{eq:IVP_independence_model}
\dot{p}_{ij}=\sum_{k\in [r_1]\setminus\{i\}}\sum_{\ell\in[r_2]\setminus\{j\}} (p_{i\ell}p_{kj}-p_{ij}p_{k\ell}),\quad p_{ij}(0)=\bar{u}_{ij},\quad (i,j)\in[r_1]\times [r_2].
\end{equation}
The unique solution to \eqref{eq:IVP_independence_model} is given explicitly by
\begin{equation}\label{eq:traj_indep}
    p_{ij}(t)=\frac{1}{u_{++}^2}\Big(u_{i+}u_{+j}+(u_{ij}u_{++}-u_{i+}u_{+j})e^{-t}\Big),
\end{equation}
where 
\[u_{+j}\coloneqq\sum_{i=1}^{r_1} u_{ij},\quad u_{i+}\coloneqq\sum_{j=1}^{r_2} u_{ij},\quad  \text{and}\quad u_{++}\coloneqq\sum_{i=1}^{r_1}\sum_{j=1}^{r_2} u_{ij} .\]
\end{prop}

\begin{proof}
Observe that 
\[ \lim_{t\to\infty} p_{ij}(t)=\frac{u_{i+}u_{+j}}{u_{++}^2}, \] 
which is consistent with the known MLE \eqref{eq:MLE_independence_model} for independence models. Furthermore, we have $p_{ij}(0)=\frac{u_{ij}}{u_{++}} = \bar{u}_{ij}$, so it suffices to show that substituting \eqref{eq:traj_indep} into the right-hand side of the ODEs \eqref{eq:IVP_independence_model}, recovers the derivatives $\dot{p}_{ij}$, which are given by
\[\dot{p}_{ij}=\frac{(u_{i+}u_{+j}-u_{ij}u_{++})}{u_{++}^2}e^{-t}.\]
The proof proceeds by direct computation, involving
the identities
\[ \sum_{k\in [r_1]\setminus\{i\}}\sum_{\ell\in[r_2]\setminus\{j\}} u_{kl} = u_{++}-u_{i+}-u_{+j}+u_{ij} \]
and
\[u_{i+}u_{+j}-u_{ij}u_{++}=\sum_{k\in [r_1]\setminus\{i\}}\sum_{\ell\in[r_2]\setminus\{j\}} (u_{i\ell}u_{kj}-u_{ij}u_{k\ell}).\qedhere\]
\end{proof}

\section{Conclusion and outlook}
\label{sec:discussion}

In this work, we describe a method that, for a given statistical log-affine model $\mc{M}_{\mm A,\vv c}$, returns a mass-action system $G_{\Lambda,\vv c}$ with the property that the MLE $\widehat{\vv p}$ of $\mc{M}_{\mm A,\vv c}$  with respect to an observed distribution $\bar{\vv u}$ is the unique positive steady state of $G_{\Lambda,\vv c}$ in the stoichiometric compatibility class of $\bar{\vv{u}}$.

Throughout the paper, we see examples of how $G_{\Lambda,\vv{c}}$ is dynamically more well-behaved if $\Lambda$ is chosen as a Markov basis for the design matrix $\mm{A}$. 
In particular, for a Markov basis, there are no relevant boundary steady states, which guarantees global convergence to the MLE.
Taking $\Lambda$ to be a $\mathbb{Z}$-basis is not enough to rule out the existence of relevant boundary steady states, as shown in~\Cref{ex:boundary_steady_states}. Furthermore, the numerical experiments presented in \Cref{sec:slow-mfld} suggest that convergence to the MLE is faster when using a Markov basis than a vector space basis. For example, for independence models, the minimal Markov bases always give global rates of convergence $1$, while with a vector space basis, the system has components whose time-scales are as high as  $10^3$, which could lead to unreliable computational performance.

From a synthetic biology point of view, these observations suggest an interesting trade-off between desirable dynamics and feasibility of chemically implementing a given network, since a Markov basis is typically much larger than a vector space basis for $\ker(\mm{A})$. 
For instance, in the independence model case, the smallest Markov basis for  $\mm{A}_\mathrm{ind}(r_1,r_2)$ has $\binom{r_1}{2}\binom{r_2}{2}$ elements, while the dimension of the kernel is $(r_1-1)(r_2-1)$. An interesting direction for future research is therefore to find systematic ways of choosing spanning sets $\Lambda$ that are not Markov bases, but still give rise to well-behaved dynamics, similarly to what was achieved in \Cref{ex:boundary_steady_states}\ref{it:vector_space_basis_with_persistence}.

\subsection*{Acknowledgments}
We thank Maize Curiel for helpful discussions in the early phases of the project, Abhishek Deshpande for sharing his expertise on siphon theory, Lucian Smith for support with the \texttt{Tellurium} package~\cite{ChoiMedleyKonigStockingEtAl2018}, and Sebastian Walcher for discussions related to rate of convergence and time-scale separations. We also thank Elisenda Feliu and Janike Oldekop for feedback on earlier versions of the manuscript, as well as the anonymous referees for several insightful suggestions.

The research by JIR was partially supported by the Alfred P. Sloan Foundation. OH was partially funded by Novo Nordisk (project NNF20OC0065582), as well as the European Union (Grant Agreement no.~101044561, POSALG). Views and opinions expressed are those of the authors only and do not necessarily reflect those of the European Union or European Research Council (ERC). Neither the European Union nor ERC can be held responsible for them. Part of this research was performed while the authors were visiting the Institute for Mathematical and Statistical Innovation (IMSI), which is supported by the National Science Foundation (Grant No.~DMS-1929348). 


\bibliographystyle{amsalpha} 
\bibliography{main}

\begin{bibdiv}
\begin{biblist}

\bib{AmendolaBlissBurkeGibbonsEtAl2019}{article}{
      author={Am{\'e}ndola, Carlos},
      author={Bliss, Nathan},
      author={Burke, Isaac},
      author={Gibbons, Courtney~R.},
      author={Helmer, Martin},
      author={Ho{\c s}ten, Serkan},
      author={Nash, Evan~D.},
      author={Rodriguez, Jose~Israel},
      author={Smolkin, Daniel},
       title={The maximum likelihood degree of toric varieties},
        date={2019},
     journal={Journal of Symbolic Computation},
      volume={92},
       pages={222\ndash 242},
}

\bib{Angeli2007}{article}{
      author={Angeli, David},
      author={De~Leenheer, Patrick},
      author={Sontag, Eduardo},
       title={A {P}etri net approach to persistence analysis in chemical reaction networks},
        date={2007},
     journal={Biology and control theory: current challenges},
       pages={181\ndash 216},
}

\bib{Almendrahernández2024}{misc}{
      author={Almendra-Hernández, Félix},
      author={Loera, Jesús A.~De},
      author={Petrović, Sonja},
       title={Markov bases: a 25 year update},
        date={2024},
        note={Preprint: \texttt{2306.06270v3}},
}

\bib{AndersonJoshi2024}{misc}{
      author={Anderson, David~F.},
      author={Joshi, Badal},
       title={Chemical mass-action systems as analog computers: implementing arithmetic computations at specified speed},
        date={2024},
      number={arXiv:2404.04396},
        note={Preprint: \texttt{arXiv:2404.04396v1}},
}

\bib{AndersonJoshiDeshpande2021}{article}{
      author={Anderson, David~F.},
      author={Joshi, Badal},
      author={Deshpande, Abhishek},
       title={On reaction network implementations of neural networks},
        date={2021},
     journal={Journal of The Royal Society Interface},
      volume={18},
      number={177},
       pages={rsif.2021.0031, 20210031},
}

\bib{AKK2020}{article}{
      author={Am\'endola, Carlos},
      author={Kosta, Dimitra},
      author={Kubjas, Kaie},
       title={Maximum likelihood estimation of toric {F}ano varieties},
        date={2020},
        ISSN={2693-2997,2693-3004},
     journal={Algebr. Stat.},
      volume={11},
      number={1},
       pages={5\ndash 30},
         url={https://doi.org/10.2140/astat.2020.11.5},
      review={\MR{4196403}},
}

\bib{invariant-toric}{article}{
      author={Am{\'e}ndola, Carlos},
      author={Kohn, Kathl{\'e}n},
      author={Reichenbach, Philipp},
      author={Seigal, Anna},
       title={Toric invariant theory for maximum likelihood estimation in log-linear models},
        date={2021},
     journal={Algebraic Statistics},
      volume={12},
      number={2},
       pages={187\ndash 211},
}

\bib{Anderson2011}{article}{
      author={Anderson, David~F.},
       title={A {{Proof}} of the {{Global Attractor Conjecture}} in the {{Single Linkage Class Case}}},
        date={2011},
     journal={SIAM Journal on Applied Mathematics},
      volume={71},
      number={4},
       pages={1487\ndash 1508},
}

\bib{AlRadhawiAngeli16}{article}{
      author={Al-Radhawi, Muhammad~Ali},
      author={Angeli, David},
       title={New approach to the stability of chemical reaction networks: piecewise linear in rates {Lyapunov} functions},
        date={2016},
        ISSN={0018-9286},
     journal={IEEE Transactions on Automatic Control},
      volume={61},
      number={1},
       pages={76\ndash 89},
}

\bib{BorosHofbauer2020}{article}{
      author={Boros, Bal{\'a}zs},
      author={Hofbauer, Josef},
       title={Permanence of {{Weakly Reversible Mass-Action Systems}} with a {{Single Linkage Class}}},
        date={2020},
     journal={SIAM Journal on Applied Dynamical Systems},
      volume={19},
      number={1},
       pages={352\ndash 365},
}

\bib{BCS2022}{article}{
      author={Brustenga~i Moncus\'i, Laura},
      author={Craciun, Gheorghe},
      author={Sorea, Miruna-\c{S}tefana},
       title={Disguised toric dynamical systems},
        date={2022},
        ISSN={0022-4049,1873-1376},
     journal={Journal of Pure and Applied Algebra},
      volume={226},
      number={8},
       pages={Paper No. 107035, 24},
         url={https://doi.org/10.1016/j.jpaa.2022.107035},
      review={\MR{4379914}},
}

\bib{Boltzmann_1887}{article}{
      author={Boltzmann, Ludwig},
       title={{Neuer Beweis zweier S\"atze \"uber das W\"armegleich-gewicht unter mehratomigen Gasmolek\"ulen}},
        date={1887},
     journal={Sitzungsberichte der Kaiserlichen Akademie der Wissenschaften},
      volume={95},
       pages={pp. 153\ndash 164},
}

\bib{Boltzmann_1896}{book}{
      author={Boltzmann, Ludwig},
       title={{Gastheorie}},
   publisher={J. A. Barth, Leipzig},
        date={1896},
}

\bib{BuismanEikelderHilbersLiekens}{article}{
      author={Buisman, H.J.},
      author={{ten Eikelder}, Huub~M.M.},
      author={Hilbers, Peter~A.J.},
      author={Liekens, Anthony~M.L.},
       title={Computing algebraic functions with biochemical reaction networks},
        date={2009},
        ISSN={1064-5462},
     journal={Artificial Life},
      volume={15},
      number={1},
       pages={5\ndash 19},
}

\bib{ChenDotyReevesSoloveichik2023}{article}{
      author={Chen, Ho-Lin},
      author={Doty, David},
      author={Reeves, Wyatt},
      author={Soloveichik, David},
       title={Rate-independent computation in continuous chemical reaction networks},
        date={2023},
     journal={Journal of the ACM},
      volume={70},
      number={3},
       pages={1\ndash 61},
}

\bib{CDSS09}{article}{
      author={Craciun, Gheorghe},
      author={Dickenstein, Alicia},
      author={Shiu, Anne},
      author={Sturmfels, Bernd},
       title={Toric dynamical systems},
        date={2009},
        ISSN={0747-7171,1095-855X},
     journal={Journal of Symbolic Computation},
      volume={44},
      number={11},
       pages={1551\ndash 1565},
         url={https://doi.org/10.1016/j.jsc.2008.08.006},
      review={\MR{2561288}},
}

\bib{ChenDalchauSrinivasPhillipsCardelliSoloveichikSeelig2013}{article}{
      author={Chen, Yuan-Jyue},
      author={Dalchau, Neil},
      author={Srinivas, Niranjan},
      author={Phillips, Andrew},
      author={Cardelli, Luca},
      author={Soloveichik, David},
      author={Seelig, Georg},
       title={Programmable chemical controllers made from {DNA}},
        date={2013},
     journal={Nature nanotechnology},
      volume={8},
      number={10},
       pages={755\ndash 762},
}

\bib{ChiconeODE}{book}{
      author={Chicone, Carmen},
       title={Ordinary differential equations with applications},
      series={Texts in Applied Mathematics},
   publisher={Springer Cham},
        date={2024},
}

\bib{ConradiIosifKahle2019}{article}{
      author={Conradi, Carsten},
      author={Iosif, Alexandru},
      author={Kahle, Thomas},
       title={Multistationarity in the space of total concentrations for systems that admit a monomial parametrization},
        date={2019},
        ISSN={0092-8240},
     journal={Bulletin of Mathematical Biology},
      volume={81},
      number={10},
       pages={4174\ndash 4209},
}

\bib{CJS2025}{article}{
      author={Craciun, Gheorghe},
      author={Jin, Jiaxin},
      author={Sorea, Miruna-\c{S}tefana},
       title={The structure of the toric locus of a reaction network},
        date={2025},
        ISSN={0951-7715,1361-6544},
     journal={Nonlinearity},
      volume={38},
      number={1},
       pages={Paper No. 015023, 28},
         url={https://doi.org/10.1088/1361-6544/ad9c0b},
      review={\MR{4844581}},
}

\bib{CraciunJinYu2020}{article}{
      author={Craciun, Gheorghe},
      author={Jin, Jiaxin},
      author={Yu, Polly~Y.},
       title={An efficient characterization of complex-balanced, detailed-balanced, and weakly reversible systems},
        date={2020},
     journal={SIAM Journal on Applied Mathematics},
      volume={80},
      number={1},
       pages={183\ndash 205},
}

\bib{cox2011toric}{book}{
      author={Cox, David},
      author={Little, John},
      author={Schenck, H.K.},
       title={{T}oric {V}arieties},
      series={Graduate studies in mathematics},
   publisher={American Mathematical Society},
        date={2011},
        ISBN={9780821848197},
         url={https://books.google.se/books?id=AoSDAwAAQBAJ},
}

\bib{ChoiMedleyKonigStockingEtAl2018}{article}{
      author={Choi, Kiri},
      author={Medley, J.~Kyle},
      author={K{\"o}nig, Matthias},
      author={Stocking, Kaylene},
      author={Smith, Lucian},
      author={Gu, Stanley},
      author={Sauro, Herbert~M.},
       title={Tellurium: {{An}} extensible python-based modeling environment for systems and synthetic biology},
        date={2018},
     journal={Biosystems},
      volume={171},
       pages={74\ndash 79},
}

\bib{CraciunNazarovPantea2013}{article}{
      author={Craciun, Gheorghe},
      author={Nazarov, Fedor},
      author={Pantea, Casian},
       title={Persistence and {{Permanence}} of {{Mass-Action}} and {{Power-Law Dynamical Systems}}},
        date={2013},
     journal={SIAM Journal on Applied Mathematics},
      volume={73},
      number={1},
       pages={305\ndash 329},
}

\bib{CappellettiOrtizAndersonWinfree2020}{article}{
      author={Cappelletti, Daniele},
      author={Ortiz-Mu{\~n}oz, Andr{\'e}s},
      author={Anderson, David~F.},
      author={Winfree, Erik},
       title={Stochastic chemical reaction networks for robustly approximating arbitrary probability distributions},
        date={2020},
        ISSN={0304-3975},
     journal={Theoretical Computer Science},
      volume={801},
       pages={64\ndash 95},
}

\bib{Craciun2019}{article}{
      author={Craciun, Gheorghe},
       title={Polynomial {{Dynamical Systems}}, {{Reaction Networks}}, and {{Toric Differential Inclusions}}},
        date={2019},
     journal={SIAM Journal on Applied Algebra and Geometry},
      volume={3},
      number={1},
       pages={87\ndash 106},
}

\bib{CappellettiWiuf2016}{article}{
      author={Cappelletti, Daniele},
      author={Wiuf, Carsten},
       title={Product-form {Poisson}-like distributions and complex balanced reaction systems},
        date={2016},
        ISSN={0036-1399},
     journal={SIAM Journal on Applied Mathematics},
      volume={76},
      number={1},
       pages={411\ndash 432},
}

\bib{DDPS2023}{article}{
      author={Davies, Isobel},
      author={Duarte, Eliana},
      author={Portakal, Irem},
      author={Sorea, Miruna-\c{S}tefana},
       title={Families of polytopes with rational linear precision in higher dimensions},
        date={2023},
        ISSN={1615-3375,1615-3383},
     journal={Found. Comput. Math.},
      volume={23},
      number={6},
       pages={2151\ndash 2202},
         url={https://doi.org/10.1007/s10208-022-09583-7},
      review={\MR{4676715}},
}

\bib{Deshpande2014}{article}{
      author={Deshpande, Abhishek},
      author={Gopalkrishnan, Manoj},
       title={Autocatalysis in reaction networks},
        date={2014},
        ISSN={0092-8240},
     journal={Bulletin of Mathematical Biology},
      volume={76},
      number={10},
       pages={2570\ndash 2595},
}

\bib{Dic16}{incollection}{
      author={Dickenstein, Alicia},
       title={Biochemical reaction networks: an invitation for algebraic geometers},
        date={2016},
   booktitle={Mathematical congress of the {A}mericas. first mathematical congress of the {A}mericas, {G}uanajuato, {M}\'exico, {A}ugust 5--9, 2013},
   publisher={Providence, RI: American Mathematical Society (AMS)},
       pages={65\ndash 83},
}

\bib{DMS2021-rational-mle}{article}{
      author={Duarte, Eliana},
      author={Marigliano, Orlando},
      author={Sturmfels, Bernd},
       title={Discrete statistical models with rational maximum likelihood estimator},
        date={2021},
        ISSN={1350-7265,1573-9759},
     journal={Bernoulli},
      volume={27},
      number={1},
       pages={135\ndash 154},
         url={https://doi.org/10.3150/20-BEJ1231},
      review={\MR{4177364}},
}

\bib{DiaconisSturmfels1998}{article}{
      author={Diaconis, Persi},
      author={Sturmfels, Bernd},
       title={Algebraic algorithms for sampling from conditional distributions},
        date={1998},
        ISSN={0090-5364,2168-8966},
     journal={The Annals of Statistics},
      volume={26},
      number={1},
       pages={363\ndash 397},
         url={https://doi.org/10.1214/aos/1030563990},
      review={\MR{1608156}},
}

\bib{drton2008lectures}{book}{
      author={Drton, Mathias},
      author={Sturmfels, Bernd},
      author={Sullivant, Seth},
       title={Lectures on algebraic statistics},
   publisher={Springer Science \& Business Media},
        date={2008},
      volume={39},
}

\bib{Fei19}{book}{
      author={Feinberg, Martin},
       title={Foundations of chemical reaction network theory},
      series={Applied Mathematical Sciences},
   publisher={Springer International Publishing},
        date={2019},
        ISBN={9783030038588},
         url={https://books.google.se/books?id=IBOGDwAAQBAJ},
}

\bib{FH24}{misc}{
      author={Feliu, Elisenda},
      author={Henriksson, Oskar},
       title={Toricity of vertically parametrized systems with applications to reaction network theory},
        date={2024},
         url={https://arxiv.org/abs/2411.15134},
        note={Preprint: \texttt{2411.15134}},
}

\bib{fienberg2007three}{article}{
      author={Fienberg, Stephen~E.},
      author={Rinaldo, Alessandro},
       title={Three centuries of categorical data analysis: Log-linear models and maximum likelihood estimation},
        date={2007},
     journal={Journal of Statistical Planning and Inference},
      volume={137},
      number={11},
       pages={3430\ndash 3445},
}

\bib{FienbergRinaldo2012}{article}{
      author={Fienberg, Stephen~E.},
      author={Rinaldo, Alessandro},
       title={Maximum likelihood estimation in log-linear models},
        date={2012},
        ISSN={0090-5364,2168-8966},
     journal={The Annals of Statistics},
      volume={40},
      number={2},
       pages={996\ndash 1023},
         url={https://doi.org/10.1214/12-AOS986},
      review={\MR{2985941}},
}

\bib{FS2025}{misc}{
      author={Feliu, Elisenda},
      author={Shiu, Anne},
       title={From chemical reaction networks to algebraic and polyhedral geometry --- and back again},
        date={2025},
        note={Preprint: \texttt{2501.06354}},
}

\bib{GatermannEiswirthSensse2005}{article}{
      author={Gatermann, Karin},
      author={Eiswirth, Markus},
      author={Sensse, Anke},
       title={Toric ideals and graph theory to analyze {H}opf bifurcations in mass action systems},
        date={2005},
        ISSN={0747-7171,1095-855X},
     journal={J. Symbolic Comput.},
      volume={40},
      number={6},
       pages={1361\ndash 1382},
         url={https://doi.org/10.1016/j.jsc.2005.07.002},
      review={\MR{2178092}},
}

\bib{GatermannHuber2002}{article}{
      author={Gatermann, Karin},
      author={Huber, Birkett},
       title={A family of sparse polynomial systems arising in chemical reaction systems},
        date={2002},
        ISSN={0747-7171},
     journal={Journal of Symbolic Computation},
      volume={33},
      number={3},
       pages={275\ndash 305},
         url={semanticscholar.org/paper/0ae459f213499293c2eb69ffd51b1477961a2c9b},
}

\bib{geiger2006toric}{article}{
      author={Geiger, Dan},
      author={Meek, Christopher},
      author={Sturmfels, Bernd},
       title={On the toric algebra of graphical models},
        date={2006},
        ISSN={0090-5364},
     journal={The Annals of Statistics},
      volume={34},
      number={3},
       pages={1463\ndash 1492},
}

\bib{GopalkrishnanMillerShiu2014}{article}{
      author={Gopalkrishnan, Manoj},
      author={Miller, Ezra},
      author={Shiu, Anne},
       title={A {{Geometric Approach}} to the {{Global Attractor Conjecture}}},
        date={2014},
     journal={SIAM Journal on Applied Dynamical Systems},
      volume={13},
      number={2},
       pages={758\ndash 797},
}

\bib{gopalkrishnan2011catalysis}{article}{
      author={Gopalkrishnan, Manoj},
       title={Catalysis in reaction networks},
        date={2011},
     journal={Bulletin of mathematical biology},
      volume={73},
      number={12},
       pages={2962\ndash 2982},
}

\bib{Gopalkrishnan2016}{incollection}{
      author={Gopalkrishnan, Manoj},
       title={A scheme for molecular computation of maximum likelihood estimators for log-linear models},
        date={2016},
   booktitle={{{DNA Computing}} and {{Molecular Programming}}},
      editor={Rondelez, Yannick},
      editor={Woods, Damien},
      volume={9818},
   publisher={Springer International Publishing},
       pages={3\ndash 18},
}

\bib{GasparToth2023}{article}{
      author={Gáspár, Vilmos},
      author={Tóth, János},
       title={Reaction extent or advancement of reaction: a definition for complex chemical reactions},
        date={2023},
        ISSN={1054-1500},
     journal={Chaos},
      volume={33},
      number={4},
       pages={22},
        note={Id/No 043141},
}

\bib{Gunawardena_review}{misc}{
      author={Gunawardena, Jeremy},
       title={Chemical reaction network theory for \emph{in-silico} biologists},
        date={2003},
        note={Available at \url{http://vcp.med.harvard.edu/papers/crnt.pdf}},
}

\bib{GatermannWolfrum2005}{article}{
      author={Gatermann, Karin},
      author={Wolfrum, Matthias},
       title={Bernstein's second theorem and {V}iro's method for sparse polynomial systems in chemistry},
        date={2005},
        ISSN={0196-8858,1090-2074},
     journal={Advances in Applied Mathematics},
      volume={34},
      number={2},
       pages={252\ndash 294},
         url={https://doi.org/10.1016/j.aam.2004.04.003},
      review={\MR{2110552}},
}

\bib{haber1985maximum}{article}{
      author={Haber, Michael},
       title={Maximum likelihood methods for linear and log-linear models in categorical data},
        date={1985},
     journal={Computational Statistics \& Data Analysis},
      volume={3},
       pages={1\ndash 10},
}

\bib{HERNANDEZBERMEJO199531}{article}{
      author={Hernández-Bermejo, Benito},
      author={Fairén, Víctor},
       title={Nonpolynomial vector fields under the lotka-volterra normal form},
        date={1995},
        ISSN={0375-9601},
     journal={Physics Letters A},
      volume={206},
      number={1},
       pages={31\ndash 37},
}

\bib{HJ72}{article}{
      author={Horn, Fritz},
      author={Jackson, Roy},
       title={General mass action kinetics},
        date={1972},
        ISSN={0003-9527},
     journal={Archive for Rational Mechanics and Analysis},
      volume={47},
       pages={81\ndash 116},
         url={https://doi.org/10.1007/BF00251225},
      review={\MR{400923}},
}

\bib{numpy}{article}{
      author={Harris, Charles~R.},
      author={Millman, K.~Jarrod},
      author={van~der Walt, St{\'{e}}fan~J.},
      author={Gommers, Ralf},
      author={Virtanen, Pauli},
      author={Cournapeau, David},
      author={Wieser, Eric},
      author={Taylor, Julian},
      author={Berg, Sebastian},
      author={Smith, Nathaniel~J.},
      author={Kern, Robert},
      author={Picus, Matti},
      author={Hoyer, Stephan},
      author={van Kerkwijk, Marten~H.},
      author={Brett, Matthew},
      author={Haldane, Allan},
      author={del R{\'{i}}o, Jaime~Fern{\'{a}}ndez},
      author={Wiebe, Mark},
      author={Peterson, Pearu},
      author={G{\'{e}}rard-Marchant, Pierre},
      author={Sheppard, Kevin},
      author={Reddy, Tyler},
      author={Weckesser, Warren},
      author={Abbasi, Hameer},
      author={Gohlke, Christoph},
      author={Oliphant, Travis~E.},
       title={Array programming with {NumPy}},
        date={2020},
     journal={Nature},
      volume={585},
      number={7825},
       pages={357\ndash 362},
         url={https://doi.org/10.1038/s41586-020-2649-2},
}

\bib{Horn1972}{article}{
      author={Horn, Fritz},
       title={Necessary and sufficient conditions for complex balancing in chemical kinetics},
        date={1972},
     journal={Archive for Rational Mechanics and Analysis},
      volume={49},
       pages={172\ndash 186},
}

\bib{Hor74}{incollection}{
      author={Horn, Fritz},
       title={The dynamics of open reaction systems},
        date={1974},
   booktitle={Mathematical aspects of chemical and biochemical problems and quantum chemistry ({P}roc. {SIAM}-{AMS} {S}ympos. {A}ppl. {M}ath., {N}ew {Y}ork, 1974)},
      series={SIAM-AMS Proc.},
      volume={Vol. VIII},
   publisher={American Mathematical Society},
       pages={125\ndash 137},
      review={\MR{468667}},
}

\bib{Likelihood-geometry-survey}{incollection}{
      author={Huh, June},
      author={Sturmfels, Bernd},
       title={Likelihood geometry},
        date={2014},
   booktitle={Combinatorial algebraic geometry},
      series={Lecture Notes in Math.},
      volume={2108},
   publisher={Springer, Cham},
       pages={63\ndash 117},
         url={https://doi-org.ezproxy.library.wisc.edu/10.1007/978-3-319-04870-3_3},
      review={\MR{3329087}},
}

\bib{hara2012hierarchical}{article}{
      author={Hara, Hisayuki},
      author={Sei, Tomonari},
      author={Takemura, Akimichi},
       title={Hierarchical subspace models for contingency tables},
        date={2012},
     journal={Journal of Multivariate Analysis},
      volume={103},
      number={1},
       pages={19\ndash 34},
}

\bib{huh2014-mldegree-one}{article}{
      author={Huh, June},
       title={Varieties with maximum likelihood degree one},
        date={2014},
     journal={Journal of Algebraic Statistics},
      volume={5},
      number={1},
}

\bib{Kerner1981-recasting}{article}{
      author={Kerner, Edward~H.},
       title={Universal formats for nonlinear ordinary differential systems},
        date={1981},
        ISSN={0022-2488},
     journal={Journal of Mathematical Physics},
      volume={22},
      number={7},
       pages={1366\ndash 1371},
         url={https://doi.org/10.1063/1.525074},
}

\bib{lauritzen1996graphical}{book}{
      author={Lauritzen, Steffen~L.},
       title={Graphical models},
   publisher={Clarendon Press},
        date={1996},
      volume={17},
}

\bib{LvLiShiFanWang2021}{article}{
      author={Lv, Hui},
      author={Li, Qian},
      author={Shi, Jiye},
      author={Fan, Chunhai},
      author={Wang, Fei},
       title={Biocomputing based on {DNA} strand displacement reactions},
        date={2021},
     journal={ChemPhysChem},
      volume={22},
      number={12},
       pages={1151\ndash 1166},
}

\bib{Marcondes2017}{article}{
      author={Marcondes~de Freitas, Michael},
      author={Feliu, Elisenda},
      author={Wiuf, Carsten},
       title={Intermediates, catalysts, persistence, and boundary steady states},
        date={2017},
     journal={Journal of Mathematical Biology},
      volume={74},
       pages={887\ndash 932},
}

\bib{MillanDickensteinShiuConradi2012}{article}{
      author={Mill{\'a}n, Mercedes~P{\'e}rez},
      author={Dickenstein, Alicia},
      author={Shiu, Anne},
      author={Conradi, Carsten},
       title={Chemical reaction systems with toric steady states},
        date={2012},
        ISSN={0092-8240},
     journal={Bulletin of Mathematical Biology},
      volume={74},
      number={5},
       pages={1027\ndash 1065},
         url={hdl.handle.net/11336/19942},
}

\bib{MullerRegensburger2012}{article}{
      author={M\"{u}ller, Stefan},
      author={Regensburger, Georg},
       title={Generalized mass action systems: Complex balancing equilibria and sign vectors of the stoichiometric and kinetic-order subspaces},
        date={2012},
     journal={SIAM Journal on Applied Mathematics},
      volume={72},
      number={6},
       pages={1926\ndash 1947},
}

\bib{miller2005combinatorial}{book}{
      author={Miller, Ezra},
      author={Sturmfels, Bernd},
       title={Combinatorial commutative algebra},
   publisher={Springer Science \& Business Media},
        date={2005},
      volume={227},
}

\bib{MichalekSturmfels2021}{book}{
      author={Micha{\l}ek, Mateusz},
      author={Sturmfels, Bernd},
       title={Invitation to nonlinear algebra},
   publisher={American Mathematical Society},
        date={2021},
      volume={211},
}

\bib{Pantea2012}{article}{
      author={Pantea, Casian},
       title={On the {{Persistence}} and {{Global Stability}} of {{Mass-Action Systems}}},
        date={2012},
     journal={SIAM Journal on Mathematical Analysis},
      volume={44},
      number={3},
       pages={1636\ndash 1673},
}

\bib{Plesa2023}{article}{
      author={Plesa, Tomislav},
       title={Stochastic approximations of higher-molecular by bi-molecular reactions},
        date={2023},
     journal={Journal of Mathematical Biology},
      volume={86},
      number={2},
}

\bib{MillanDickenstein2018}{article}{
      author={P{\'e}rez~Mill{\'a}n, Mercedes},
      author={Dickenstein, Alicia},
       title={The structure of {MESSI} biological systems},
        date={2018},
        ISSN={1536-0040},
     journal={SIAM Journal on Applied Dynamical Systems},
      volume={17},
      number={2},
       pages={1650\ndash 1682},
}

\bib{PachterSturmfels2005}{book}{
      author={Pachter, Lior},
      author={Sturmfels, Bernd},
       title={Algebraic statistics for computational biology},
   publisher={Cambridge university press},
        date={2005},
      volume={13},
}

\bib{QianSoloveichikWinfree2011}{incollection}{
      author={Qian, Lulu},
      author={Soloveichik, David},
      author={Winfree, Erik},
       title={Efficient {Turing}-universal computation with {DNA} polymers},
        date={2011},
   booktitle={{DNA} computing and molecular programming. 16th international conference, {DNA} 16, {H}ong {K}ong, {C}hina, june 14--17, 2010.},
   publisher={Berlin: Springer},
       pages={123\ndash 140},
         url={resolver.caltech.edu/CaltechAUTHORS:20111006-081731222},
}

\bib{Johnston2008}{article}{
      author={Siegal, David},
      author={Johnston, Matthew~D.},
       title={Linearization of {{Complex Balanced Reaction Systems}}},
        date={2015},
        note={Available at \url{https://johnstonmd.wordpress.com/wp-content/uploads/2015/02/linearization.pdf}},
}

\bib{SiegelMacLean2000}{article}{
      author={Siegel, David},
      author={MacLean, Debra},
       title={Global stability of complex balanced mechanisms},
        date={2000},
     journal={Journal of Mathematical Chemistry},
      volume={27},
      number={1/2},
       pages={89\ndash 110},
}

\bib{Sontag2001}{article}{
      author={Sontag, Eduardo~D.},
       title={Structure and stability of certain chemical networks and applications to the kinetic proofreading model of {{T-cell}} receptor signal transduction},
        date={2001},
     journal={IEEE Transactions on Automatic Control},
      volume={46},
      number={7},
       pages={1028\ndash 1047},
}

\bib{Srinivas2015thesis}{book}{
      author={Srinivas, Niranjan},
       title={Programming chemical kinetics: engineering dynamic reaction networks with {DNA} strand displacement},
   publisher={California Institute of Technology},
        date={2015},
}

\bib{Shiu2010}{article}{
      author={Shiu, Anne},
      author={Sturmfels, Bernd},
       title={Siphons in chemical reaction networks},
        date={2010},
     journal={Bulletin of Mathematical Biology},
      volume={72},
       pages={1448\ndash 1463},
}

\bib{SoloveichikSeeligWinfree2010}{article}{
      author={Soloveichik, David},
      author={Seelig, Georg},
      author={Winfree, Erik},
       title={{{{DNA}}} as a universal substrate for chemical kinetics},
        date={2010},
     journal={Proceedings of the National Academy of Sciences},
      volume={107},
      number={12},
       pages={5393\ndash 5398},
}

\bib{Sul18}{book}{
      author={Sullivant, Seth},
       title={Algebraic statistics},
      series={Graduate Studies in Mathematics},
   publisher={American Mathematical Society},
        date={2018},
        ISBN={9781470435172},
}

\bib{4ti2}{misc}{
      author={4ti2 team},
       title={4ti2---a software package for algebraic, geometric and combinatorial problems on linear spaces},
         url={https://4ti2.github.io},
}

\bib{VirinchiBeheraGopalkrishnan2017}{inproceedings}{
      author={Virinchi, Muppirala~Viswa},
      author={Behera, Abhishek},
      author={Gopalkrishnan, Manoj},
       title={A stochastic molecular scheme for an artificial cell to infer its environment from partial observations},
        date={2017},
   booktitle={{DNA} computing and molecular programming},
      editor={Brijder, Robert},
      editor={Qian, Lulu},
   publisher={Springer International Publishing},
     address={Cham},
       pages={82\ndash 97},
}

\bib{scipy}{article}{
      author={Virtanen, Pauli},
      author={Gommers, Ralf},
      author={Oliphant, Travis~E.},
      author={Haberland, Matt},
      author={Reddy, Tyler},
      author={Cournapeau, David},
      author={Burovski, Evgeni},
      author={Peterson, Pearu},
      author={Weckesser, Warren},
      author={Bright, Jonathan},
      author={{van der Walt}, St{\'e}fan~J.},
      author={Brett, Matthew},
      author={Wilson, Joshua},
      author={Millman, K.~Jarrod},
      author={Mayorov, Nikolay},
      author={Nelson, Andrew R.~J.},
      author={Jones, Eric},
      author={Kern, Robert},
      author={Larson, Eric},
      author={Carey, C~J},
      author={Polat, {\.I}lhan},
      author={Feng, Yu},
      author={Moore, Eric~W.},
      author={{VanderPlas}, Jake},
      author={Laxalde, Denis},
      author={Perktold, Josef},
      author={Cimrman, Robert},
      author={Henriksen, Ian},
      author={Quintero, E.~A.},
      author={Harris, Charles~R.},
      author={Archibald, Anne~M.},
      author={Ribeiro, Ant{\^o}nio~H.},
      author={Pedregosa, Fabian},
      author={{van Mulbregt}, Paul},
      author={{SciPy 1.0 Contributors}},
       title={{{SciPy} 1.0: Fundamental Algorithms for Scientific Computing in Python}},
        date={2020},
     journal={Nature Methods},
      volume={17},
       pages={261\ndash 272},
}

\bib{Volpert1973}{article}{
      author={Vol'pert, Aizik~I.},
       title={Differential equations on graphs},
        date={1972},
        ISSN={0025-5734},
     journal={Mathematics of the USSR, Sbornik},
      volume={17},
       pages={571\ndash 582},
}

\bib{ViswaBeheraGopalkrishnan2018}{inproceedings}{
      author={Viswa~Virinchi, Muppirala},
      author={Behera, Abhishek},
      author={Gopalkrishnan, Manoj},
       title={A reaction network scheme which implements the em algorithm},
        date={2018},
   booktitle={{DNA} computing and molecular programming},
      editor={Doty, David},
      editor={Dietz, Hendrik},
   publisher={Springer International Publishing},
     address={Cham},
       pages={189\ndash 207},
}

\bib{WiufBeheraSinghGopalkrishnan2023}{article}{
      author={Wiuf, Carsten},
      author={Behera, Abhishek},
      author={Singh, Abhinav},
      author={Gopalkrishnan, Manoj},
       title={A reaction network scheme for hidden markov model parameter learning},
        date={2023},
     journal={Journal of the Royal Society Interface},
      volume={20},
      number={203},
       pages={20220877},
}

\bib{Wilhelm2000}{article}{
      author={Wilhelm, Thomas},
       title={Chemical systems consisting only of elementary steps -- a paradigma for nonlinear behavior},
        date={2000},
     journal={Journal of Mathematical Chemistry},
      volume={27},
      number={1-2},
       pages={71\ndash 88},
}

\bib{WangLinSontagSorger2019}{article}{
      author={Wang, Shu},
      author={Lin, Jia-Ren},
      author={Sontag, Eduardo~D.},
      author={Sorger, Peter~K.},
       title={Inferring reaction network structure from single-cell, multiplex data, using toric systems theory},
        date={2019},
     journal={PLOS Computational Biology},
      volume={15},
      number={12},
       pages={1\ndash 25},
}

\bib{YuCraciun2018}{article}{
      author={Yu, Polly~Y.},
      author={Craciun, Gheorghe},
       title={Mathematical analysis of chemical reaction systems},
        date={2018},
     journal={Israel Journal of Chemistry},
      volume={58},
      number={6-7},
       pages={733\ndash 741},
}

\end{biblist}
\end{bibdiv}

\end{document}